\documentclass[11pt]{article}
\usepackage[noheader]{styles/redhat_arxiv}
\shadelogopath{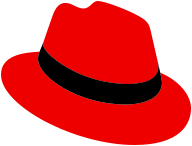}
\tcbset{shade titlebox/.append style={bottom=1.35cm}}
\newcommand{\arxivversion}{1}
\newcommand{\serveguardopenscience}{content/open_science_arxiv}
\usepackage{adjustbox}

\usepackage{enumitem}
\setlist[description]{style=nextline,font=\normalfont}

\ifdefined\arxivversion\else
  \documentclass[11pt]{article}
  \usepackage[margin=1in]{geometry}
\fi
\usepackage{amsmath,amssymb,amsthm}
\usepackage{mathtools}
\usepackage{booktabs}
\usepackage{array}
\usepackage{microtype}
\usepackage{xcolor}
\usepackage{tikz}
\usetikzlibrary{arrows.meta, positioning}
\usepackage{styles/pgfplots_academic}  %
\pgfplotscreateplotcyclelist{oi8}{
  {OIblue,  solid,          mark=*,         mark size=1.8pt, mark options={solid}},
  {OIorange,densely dashed, mark=square*,   mark size=1.7pt, mark options={solid}},
  {OIsky,   dashdotdotted,  mark=triangle*, mark size=2pt,   mark options={solid}},
  {OIgreen, densely dotted, mark=diamond*,  mark size=1.9pt, mark options={solid}},
  {OIverm,  loosely dashed, mark=otimes*,   mark size=1.9pt, mark options={solid}},
  {OIpurp,  dashdotted,     mark=pentagon*, mark size=2pt,   mark options={solid}},
  {OIgrey,  densely dashdotted, mark=x,     mark size=2pt,   mark options={solid}},
  {OIyellow,loosely dotted, mark=star,      mark size=2pt,   mark options={solid}},
}
\usepackage[sort&compress,numbers]{natbib}
\usepackage[hidelinks]{hyperref}
\usepackage{subcaption}
\usepackage[capitalize]{cleveref}

\newtheorem{theorem}{Theorem}
\newtheorem{lemma}{Lemma}

\newtheorem{proposition}{Proposition}
\newtheorem{definition}{Definition}

\newcommand{\Real}{\mathbb{R}}
\newcommand{\Fp}{\mathbb{F}_p}
\providecommand{\rank}{\operatorname{rank}}
\DeclareMathOperator{\img}{im}
\DeclareMathOperator{\Inert}{Inert}
\newcommand{\Mon}{M}              %
\newcommand{\Val}{V}              %
\newcommand{\Valbase}{V_{0}}      %
\newcommand{\Valupd}{\Delta}      %
\newcommand{\Valsrv}{V_{\mathrm{srv}}}  %
\newcommand{\pip}{\pi}            %
\newcommand{\dmodel}{d_{\mathrm{model}}}
\newcommand{\dhead}{d}            %
\newcommand{\Nbasis}{N}           %
\newcommand{\commit}{c}           %
\newcommand{\Pivis}{\Pi_{\mathrm{vis}}}   %
\newcommand{\eps}{\varepsilon}
\newcommand{\Mdep}{M_{\mathrm{dep}}}      %

\newcommand{\gptDmodel}{768}
\newcommand{\gptDhead}{64}
\newcommand{\gptBlindDim}{640}

\newcommand{\gptSeparationOrders}{7.7}
\newcommand{\gptDenseSecret}{49152}
\newcommand{\gptLoraSecretRankEight}{6656}

\newcommand{\dfourTrainedBlind}{88}
\newcommand{\dfourRandomBlind}{91}

\newcommand{\multiModelsCost}{4}
\newcommand{\multiTrainedBlindLo}{83}
\newcommand{\multiTrainedBlindHi}{88}
\newcommand{\gqaValueOutDim}{128}

\newcommand{\wvSigmaMax}{8.0}
\newcommand{\gainResidFactor}{13}
\newcommand{\gainVsQk}{5}
\newcommand{\qkFlagshipResid}{6.1}
\newcommand{\qkKerpiRank}{640}
\newcommand{\qkAugDirs}{160}
\newcommand{\qkAugResid}{1.8}

\newcommand{\zkProveHead}{6.4}
\newcommand{\zkProveLayer}{6.4}
\newcommand{\zkVerifyHeadMs}{95}
\newcommand{\zkVerifyMs}{96}
\newcommand{\zkProofHeadKB}{5.5}
\newcommand{\zkProofLayerKB}{5.5}

\newcommand{\poneCheapRankLo}{4}
\newcommand{\poneCheapRecoveryLo}{100}
\newcommand{\poneCheapParamRatio}{11.6}

\newcommand{\garakGlitchTokens}{141}
\newcommand{\garakGlitchRank}{104}
\newcommand{\garakGlitchCleanResid}{5.3\times10^{-7}}
\newcommand{\garakGlitchBackdoorResid}{2.3}
\newcommand{\garakGlitchSepOrders}{6.6}

\newcommand{\garakInjBlindFrac}{0.90}

\newcommand{\garakInjSepOrders}{6.3}

\newcommand{\quantIntEightHeadroom}{2.0}
\newcommand{\quantIntFourMarginLo}{1.5}

\newcommand{\behavShiftRatio}{67}
\newcommand{\garakBaseGlitchFail}{95}
\newcommand{\garakBaseGlitchProbes}{256}

\newcommand{\driftObvsBlindResid}{2.2}

\newcommand{\qwenInjSepOrders}{7.3}

\newcommand{\smollmInjSepOrders}{6.9}

\newcommand{\covMhaOverlap}{1.00}

\newcommand{\covGqaBlindShared}{384}
\newcommand{\behavSeeds}{30}
\newcommand{\behavRatioMean}{74}
\newcommand{\behavRatioStd}{14}
\newcommand{\steerStealthConcealed}{88}

\newcommand{\steerEvadeVisibility}{0.54}
\newcommand{\steerNaiveVisibility}{0.26}
\newcommand{\steerZkSepOrders}{4.4}

\newcommand{\detAucStealth}{0.65}
\newcommand{\detAucUnit}{0.96}
\newcommand{\advResCapLo}{0.5}
\newcommand{\advResCapHi}{1.1}

\newcommand{\realAdapterBlindFrac}{87}

\newcommand{\realAdapterSepOrders}{6.2}

\newcommand{\intManipVsBase}{31}
\newcommand{\intTamperResid}{0.98}
\newcommand{\intConfinedResid}{1.4\times10^{-6}}

\newcommand{\zkTypedProveLayer}{6.4}
\newcommand{\zkTypedVerifyMs}{96}
\newcommand{\zkTypedProofKB}{5.5}
\newcommand{\zkSaltCommitProve}{48}
\newcommand{\zkSaltCommitVerifyMs}{347}
\newcommand{\zkSaltCommitProofKB}{6.0}
\newcommand{\zkTypedTamperY}{99.7}
\newcommand{\preprocGptMs}{86}
\newcommand{\preprocQwenFiveMs}{31}
\newcommand{\preprocSmolMs}{1688}

\newcommand{\frontierModels}{8}
\newcommand{\frontierGptKtenMed}{533}
\newcommand{\frontierGptKhundredReach}{10}
\newcommand{\frontierGptKappaPercent}{83}

\newcommand{\frontierSmolKappaPercent}{79}

\newcommand{\frontierQwenOneFiveKappaPercent}{20}

\newcommand{\frontierQwenThreeKappaPercent}{14}

\newcommand{\frontierLlamaOneKappaPercent}{27}

\newcommand{\frontierLlamaThreeKappaPercent}{36}

\newcommand{\utilSeeds}{3}
\newcommand{\doeRecovGain}{99}

\newcommand{\doeRecovRand}{98}
\newcommand{\doeRecovStdMax}{2}

\newcommand{\doeResidQkaug}{43}
\newcommand{\doeResidRand}{89}
\newcommand{\doeResidRandDeep}{60}

\newcommand{\doeEtaUtil}{0.2}
\newcommand{\doeEtaResid}{85}

\newcommand{\dollyRecovGain}{98}
\newcommand{\dollyRecovRand}{99}

\newcommand{\dollyResidGain}{0}
\newcommand{\dollyResidRand}{89}
\newcommand{\dollyResidQkaug}{43}
\newcommand{\dollyEtaUtil}{12}
\newcommand{\dollyEtaResid}{81}

\newcommand{\nwCellBits}{28}
\newcommand{\nwBoundAr}{66}

\newcommand{\nwBoundCm}{65}

\newcommand{\nwBoundDiff}{66}
\newcommand{\nwMarginWorst}{186}
\newcommand{\compGapMaxPct}{0.25}
\newcommand{\compGapMedPct}{0.16}
\newcommand{\guardAdmitMs}{31}
\newcommand{\guardOpenS}{1.2}
\newcommand{\sidecarKB}{18}

\newcommand{\pkgRootLayers}{12}
\newcommand{\pkgCommitProveS}{555}
\newcommand{\pkgCommitVerifyS}{4.1}
\newcommand{\pkgRootProofKB}{70}
\newcommand{\pkgPlainProveS}{76}
\newcommand{\pkgRootCostMs}{0.03}

\newcommand{\doeMatchedMaxDiff}{0.13}

\newcommand{\doeMatchedPairs}{18}
\newcommand{\doeTostMargin}{1}

\newcommand{\serveGuardSepOrders}{6}

\newcommand{\zkProbes}{10}
\newcommand{\zkProbeSetBits}{13}
\newcommand{\zkSoundnessBits}{130}

\newcommand{\knowChanSelectivity}{7}

\ifdefined\arxivversion
\title{ServeGuard: Verifiable, Bounded-Residual Confinement of
Operator-Invisible Channels Without Revealing the Certified Read Factor}
\author{%
  \begin{minipage}[t]{0.48\linewidth}
    \raggedright
    Dominik Dahlem\textsuperscript{*}\\
    {\normalsize Red Hat AI}\\
    {\normalsize\href{mailto:ddahlem@redhat.com}{\texttt{ddahlem@redhat.com}}}
  \end{minipage}\hfill
  \begin{minipage}[t]{0.48\linewidth}
    \raggedright
    Rui Vieira\\
    {\normalsize Red Hat AI}\\
    {\normalsize\href{mailto:rui@redhat.com}{\texttt{rui@redhat.com}}}
  \end{minipage}\\[0.25em]
  {\footnotesize\textsuperscript{*}Corresponding author.}}
\shadeaffiliation{}
\hypersetup{pdfauthor={Dominik Dahlem, Rui Vieira}}

\else
\title{\bfseries ServeGuard:\\
Verifiable, Bounded-Residual Confinement of\\
Operator-Invisible Channels Without Revealing the Certified Read Factor}
\author{SHADE Program \\ (author list withheld for review)}
\fi
\date{}

\begin{document}
\maketitle

\begin{abstract}
Third-party adapters for open-weight language models ship as opaque weight matrices; a
recipient cannot check whether an adapter hides a backdoor without trusting the publisher or
inspecting the weights, the publisher's core asset.
For one important class (payloads placed where a safety monitor is structurally blind), detection
is unsound as a defense: every detector that factors through the declared monitor is invariant
on its blind subspace, and honest and backdoored adapters overlap on every blind-subspace
statistic we evaluate, because benign adaptation uses that subspace too.
Rather than detect this channel, we make it structurally \emph{absent} and prove that we did.
The publisher builds the adapter to read the input only through directions the monitor covers
and proves this in zero knowledge, revealing nothing about the read factor it certifies. The certificate is cheap
because the expensive part, identifying the monitor's blind spot, is a deterministic function
of the \emph{public} base model, so only one linear identity is proved; the served residual is
the base model's own public floor, not a prover-chosen tolerance.
The result is \emph{ServeGuard}, a supply-chain primitive: the publisher ships a
\emph{proof-carrying adapter} whose proof lets a consumer or regulator verify, without the
certified read factor and without trusting the publisher, that the adapter carries no hidden
channel of this class relative to the declared monitor; an admission-time typing guard binds
the guarantee to the adapter bytes admitted at serving time.
Across eight checkpoints up to 7B from four families, the monitoring budget is
architectural: the measured frontier saturates at the value-path rank on grouped-query
checkpoints but not on multi-head ones. On a 0.5B model confinement is nearly free for benign adaptation, making
monitor quality the security lever.

\end{abstract}

\section{Introduction}
\label{sec:intro}

Third-party model components ship as weights, and the supply chain that moves
them (fine-tuning services, model hubs, vendor releases, and the regulators who must
vet them) has no way to check what those weights carry. A fine-tuning service returns an
adapter; a model hub serves a checkpoint; a vendor delivers a quantized release. The
recipient runs the artifact inside a safety pipeline and wants assurance that it carries
no hidden backdoor; the recipient's own stakeholders (consumers, regulators, auditors)
may require that assurance without direct access to the model. Three options exist today, and each fails. Benchmark scores say nothing about a payload that activates only on a trigger.
Full weight disclosure exposes the publisher's intellectual property, and even then a
weight matrix is not human-interpretable the way source code is. Third-party audit under
non-disclosure relocates trust to the auditor rather than removing it. The recipient is
left to trust the publisher.

\paragraph{What can hide.} We study one precisely characterized channel class.
In a transformer attention layer, each token's representation passes through three
learned projections: the \emph{query} and \emph{key} maps, which determine which tokens
attend to which, and the \emph{value} map, which determines what information is read.
Deployed safety monitors read the internal state through a fixed linear feature map: a
probing classifier scores the residual stream through learned
weights~\cite{belinkov2022probing}, a representation-engineering defense watches the
projection onto a declared refusal or harmfulness
direction~\cite{zou2023repeng,arditi2024refusal}, and an attention-pattern monitor such as
Attention Tracker~\cite{hung2025attentiontracker} post-processes the query/key features of
selected heads. The post-processing may be nonlinear; what the certificate is defined
against is the linear feature map, which is \emph{public} once the operator declares what
it watches, and a payload in the feature map's kernel leaves every downstream statistic
unchanged. A
monitor that inspects only the query/key
side has a \emph{blind subspace}: perturbations it maps to zero and therefore cannot
distinguish from no perturbation at all. Such a perturbation is invisible to every test
that queries only the monitor's output, at every payload strength, yet the value path
reads it. The channel is invisible \emph{by construction}, not by stealth, and it can be
shipped as a LoRA adapter (a small, low-rank weight update added to a pretrained model).
The security-relevant property is thus not the harm a particular payload causes but the
\emph{invisibility}: no detector that factors through the declared monitor can observe a
kernel payload, and blind-space occupancy is not a sound reference-free discriminator of
benign from payload-carrying adapters (\cref{sec:eval}), so the principled defense is to
make the channel structurally absent, not to bound the payload in hand.

\paragraph{Why detection fails.} The first instinct is to detect the channel from the
released weights. It is uninformative. The monitor's blind subspace occupies most of the
residual stream (the hidden-state vector passed between layers; $\gptBlindDim$ of $\gptDmodel$ dimensions on GPT-2~\cite{gpt2}), so any low-rank
adapter places the bulk of its computation there by geometry alone: a trained honest
adapter and a backdoored one are indistinguishable on every kernel-activity statistic we
evaluate (\cref{sec:eval}). This is not an artifact of a
weak test: backdoors can be planted so as to be provably undetectable from the weights
under explicit hardness assumptions~\cite{goldwasser2022backdoors,choudhary2026sparseback},
with latent-space
constructions resisting the known defenses~\cite{eggen2026latentback}; absent a declared
reference direction, weight-space detection has no sound target.
Crucially, the no-go forbids \emph{distinguishing} a backdoored model from a clean one
without a reference; it does not forbid \emph{certifying a structural relation relative to
a public, declared monitor}. That distinction is the opening.

\paragraph{The pivot: build it out and prove it.} Rather than try to detect the channel,
the publisher \emph{builds the adapter so the channel cannot exist} and proves it. A LoRA
adapter is a product of two small matrices: a \emph{read factor} that selects which
directions of the input the adapter accesses and a \emph{write factor} that produces the
output. If the read factor is constrained to access the input only through what the
monitor exposes, then no input in the monitor's blind subspace can reach the adapter at
all, and the channel is absent by construction, regardless of what the publisher intended.

This is the model-world analogue of proof-carrying code~\cite{necula1997pcc}: a property
you build in and then prove you followed. The identity is \emph{exact}, certifying that the
adapter adds no blind-subspace read whatsoever rather than less than a threshold, so the
served model's residual exposure is the base model's own public floor and not a
prover-chosen tolerance. It is proved over fixed-point weights with the typed factors
constructed on that encoding lattice (\cref{sec:realization}), so quantization adds no
slack, and it is cheap for a structural reason: the expensive part, identifying the blind
subspace, is a deterministic function of the \emph{public} base model, so it happens in the
clear and the publisher proves one identity that is linear in the committed secret against a
public matrix. The consumer
checks static per-layer certificates once rather than monitoring every inference.

\paragraph{Contributions.}
\begin{itemize}
\item \textbf{A typed confinement defense, certified and cost-matched}
(\cref{sec:defense,sec:certificate,sec:soundness}): rather than \emph{test} a released
adapter, the publisher \emph{builds} a monitor-typed adapter whose read factor factors
through the public monitor, certified by one linear identity in zero knowledge. The served
residual is pinned \emph{exactly} at the public base floor (\cref{lem:served-floor}). The deployed monitor augments the
operator's QK map with the value path's top-gain singular directions (the spectral optimum,
\cref{thm:augment}), whose reach across architectures we \emph{measure} (\cref{sec:eval}).
\item \textbf{Security evaluation and utility characterization}
(\cref{sec:eval}): an adaptive adversary, a targeted integrity manipulation, a
cross-architecture geometry sweep (every layer and query head of \frontierModels\ checkpoints
spanning four families up to 7B), and a factorial utility study over monitor designs, budgets, ranks, and seeds on two
instruction tasks. Confinement costs little and the cost separates cleanly from the security
benefit: utility recovers to $\doeRecovRand$--$\doeRecovGain\%$ under \emph{every} monitor
while the certified residual varies by orders of magnitude, driven to zero only by a monitor
designed for the value path's gain structure ($\eta^2$: $\doeEtaUtil\%$ vs
$\doeEtaResid\%$).
\item \textbf{A measured zero-knowledge implementation}
(\cref{sec:realization,sec:eval}): a Halo2 prover certifying value adapters in
$\zkTypedProveLayer$--$\zkSaltCommitProve$\,s with
$\zkTypedVerifyMs$--$\zkSaltCommitVerifyMs$\,ms verification and
$\zkTypedProofKB$--$\zkSaltCommitProofKB$\,KB proofs (plain typed relation to one-time
salted commit-and-prove binding), at algebraic false-accept probability
$\approx 2^{-\zkSoundnessBits}$ per certificate for a fixed commitment (the end-to-end
bound adds the binding, knowledge-soundness, and grinding terms, \cref{thm:protocol}), a
whole-checkpoint \emph{package
root} binding every layer's certificate, and an admission-time typing guard that ties the
admitted read factor to the certified one.
\end{itemize}
Together these deliver a \emph{proof-carrying adapter}: a new supply-chain primitive that
turns publisher trust into a verifiable structural guarantee a recipient checks without seeing the certified read factor. We call the
end-to-end system \emph{ServeGuard}, after the serving-time guard that binds the guarantee to
the adapter bytes admitted at serving time rather than only the bytes attested
(\cref{thm:endtoend} states the composition).
The certificate is \emph{relative to a declared public monitor}: it certifies the absence of
this linear-carrier channel class in the released update, against a committed checkpoint, for
the subspace that monitor covers: a payload riding a direction the monitor does not probe is
outside it (\cref{sec:gapcorollary}), and whether that monitor is the safety-relevant object is a
declared premise, not a proved one. Separating honest adapters from backdoors is the
frontier it opens (\cref{sec:scope}).

\section{Detection from weights, and defense by construction}
\label{sec:threat}

The introduction described the supply-chain trust gap informally. This section makes it
precise: the roles, the channel class, why detection cannot separate honest from
malicious, and what defense replaces it.

\paragraph{Roles.} A \emph{prover} (the publisher) holds an adapter $\Valupd = BA$,
commits to its read factor $A$ and type witness $C$ under a binding, hiding commitment
$\commit$, and proves a structural property of the adapter. The hiding profile is fixed
throughout: $A$, $C$, and the commitment salt stay secret; the write factor $B$ ships in the
clear with the package and is hash-bound, not hidden; the proved statement is
$A = C\Mdep$. A \emph{verifier} (consumer, hub, or
regulator) holds $\commit$ and the public statement, and either the public base model or the published
monitor matrix $\Mdep$ together with an authenticated manifest binding it to the base, and
accepts or rejects. The
\emph{adversary} is the publisher itself, shipping a covert operator-invisible channel while
passing benchmarks~\cite{badnets2017}, an untrusted fine-tuning service that injects a payload
before the adapter is committed, or a distribution intermediary that substitutes a package
(caught by the package binding of \cref{sec:realization}). A compromised \emph{serving
endpoint} is out of scope: the certificate binds the committed checkpoint, not what an
endpoint loads at inference (\cref{sec:scope}). The headline deployment is a third-party LoRA/PEFT
adapter on a \emph{public} base. By \emph{public} we mean an open-weight model whose
parameters are available to all parties (e.g.\ Llama, Qwen, GPT-2 on a model hub): the
base weights and therefore any monitor derived from them (such as the QK map $\pip$ or its
augmented variant $\Mdep$) are public and fixed, and the only untrusted object is the
adapter.
Certifying the untrusted
adapter against a trusted public base is the natural trust boundary of this deployment.
In practice, a \emph{platform} (a model hub or inference service) may provide the proving
infrastructure, with the deployer's own stakeholders verifying without weight access
(\cref{fig:pipeline}); \cref{sec:scope} states what changes when the base is not public.

\paragraph{The confinement attestation problem.}
Given a public base model with value map $\Valbase$ and a public linear monitor
$\Mon$, and an adapter $\Valupd = BA$ whose read factor is committed under a binding, hiding
$\commit$ and whose write factor $B$ ships in the clear with the package
(\cref{sec:realization}): construct a proof system in which the publisher convinces the
verifier that $\Valupd$ adds no component readable only in the monitor's blind subspace
$\ker\Mon$ (i.e.\ $\ker\Mon\subseteq\ker\Valupd$), such that completeness and soundness
hold and the transcript reveals nothing about $(A,C)$ beyond $\commit$. The served model's residual blind
response should equal the public base contribution $\Valbase$ restricted to $\ker\Mon$,
computable by the verifier in the clear.

\paragraph{The released model: base plus update.} The served \emph{value map} (the linear
projection $W_V$ inside each attention layer, which determines what information the layer
extracts from the residual stream) decomposes as
$\Valsrv = \Valbase + \Valupd$: a \emph{base} value map $\Valbase$, a deterministic public
function of the released \emph{public} base model, plus a committed
\emph{secret update} $\Valupd$ (the shipped adapter; for a rank-$r$
LoRA~\cite{lora2021}, a low-rank factorization $\Valupd = BA$, with $\commit$ binding
the read factor $A$ and the package root binding $B$ together with $\commit$,
\cref{sec:realization}). The attack surface is the update: an operator-invisible channel is what the publisher
\emph{adds} through $\Valupd$. Because the base model is public, its blind read
$\Valbase\Nbasis$ against any public basis $\Nbasis$ of $\ker\pip$ is computable by every
verifier in the clear: a property of an artifact everyone holds, not a concealed
insertion. The certificate therefore attests the update's contribution $\Valupd\Nbasis$,
pinning the served residual blind read at the public base floor $\Valbase\Nbasis$
(\cref{sec:certificate}), which a low-rank update cannot drive to zero
(\cref{prop:served}(b)).

\paragraph{Who computes the monitor.} Where the monitor is computed sets the trust model.
We target the \emph{public-base, secret-adapter} case: the base is public, so the verifier
computes the deployed monitor $\Mdep$ (the QK map plus its optimal in-kernel augmentation,
\cref{sec:defense}) from the public base, with no access to the secret adapter and no choice
left to the publisher, so it cannot be gamed to hide structure. The alternatives are weaker
or costlier: a public \emph{fixed} monitor (the base QK map) is verifier-computable but
value-path-misaligned; a \emph{publisher-proposed} monitor would itself need a ZK proof of
its relation to the committed weights; a \emph{trusted-auditor or TEE} monitor needs weight
access.

\section{The defense: confinement to the visible channel}
\label{sec:defense}

The previous section established that detection from weights is uninformative and that
the defense must be structural. This section gives the construction: confine the adapter
to the monitor's visible channel so no channel in the deployed monitor's blind subspace
$\ker\Mdep$ can exist, and choose the monitor to minimize the residual cost (\cref{fig:pipeline}).

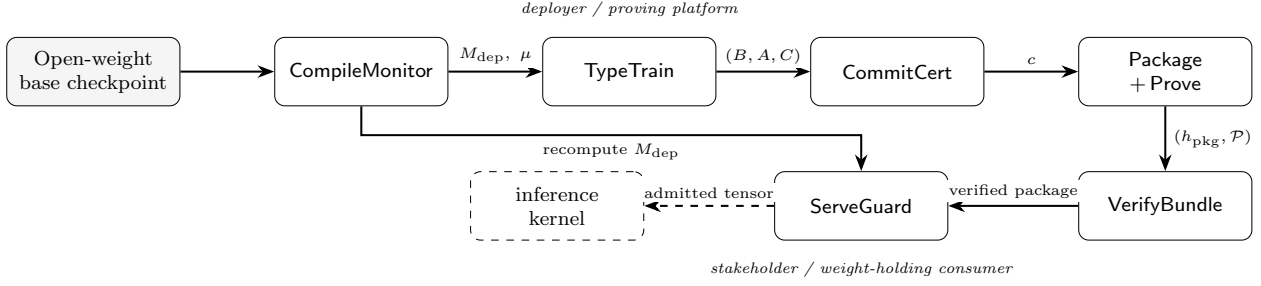
\begin{figure}[t]
\centering
\ifdefined\arxivversion\begin{adjustbox}{max width=\linewidth}\fi
\begin{tikzpicture}[>=Stealth,
  op/.style={draw, rounded corners, align=center, font=\scriptsize, inner sep=3.5pt,
             minimum height=9.5mm, minimum width=24mm},
  pub/.style={op, fill=black!4},
  lbl/.style={font=\tiny, fill=white, inner sep=1.5pt, align=center}]
\node[pub] (base) {Open-weight\\base checkpoint};
\node[op, right=13mm of base] (cm) {$\mathsf{CompileMonitor}$};
\node[op, right=13mm of cm] (tt) {$\mathsf{TypeTrain}$};
\node[op, right=13mm of tt] (cp) {$\mathsf{CommitCert}$};
\node[op, right=13mm of cp] (pkg) {$\mathsf{Package}$\\$+\,\mathsf{Prove}$};
\node[op, below=9mm of pkg] (vb) {$\mathsf{VerifyBundle}$};
\node[op, left=18mm of vb] (sg) {$\mathsf{ServeGuard}$};
\node[op, dashed, left=18mm of sg] (ik) {inference\\kernel};
\node[font=\tiny\itshape, above=1.5mm of tt] {deployer / proving platform};
\node[font=\tiny\itshape, below=1.5mm of sg] {stakeholder / weight-holding consumer};
\draw[->, thick] (base) -- (cm);
\draw[->, thick] (cm) -- node[lbl, above=0.4mm]{$\Mdep,\ \mu$} (tt);
\draw[->, thick] (tt) -- node[lbl, above=0.4mm]{$(B,A,C)$} (cp);
\draw[->, thick] (cp) -- node[lbl, above=0.4mm]{$\commit$} (pkg);
\draw[->, thick] (pkg) -- node[lbl, right=1.5pt]{$(h_{\mathrm{pkg}},\mathcal P)$} (vb);
\draw[->, thick] (vb) -- node[lbl, above=0.4mm]{verified package} (sg);
\draw[->, thick, dashed] (sg) -- node[lbl, above=0.4mm]{admitted tensor} (ik);
\coordinate (elb) at ([yshift=-4mm]cm.south);
\draw[->, thick] (cm.south) -- (elb) -- node[lbl, below=0.5mm]{recompute $\Mdep$}
  (sg.north |- elb) -- (sg.north);
\end{tikzpicture}
\ifdefined\arxivversion\end{adjustbox}\fi
\caption{The ServeGuard pipeline. The deployer or its proving platform compiles the public
monitor $\Mdep$ from the open base, trains monitor-typed factors, commits to the certified
read factor and type witness, computes the package root $h_{\mathrm{pkg}}$ binding the
public write factor to $\commit$, proves the typed identity against a statement carrying
that root, and ships the package. A stakeholder verifies without the certified read
factor; the consumer's admission-time guard recomputes $\Mdep$ and admits only bytes opening
$\commit$ as the typed identity. The dashed edge is the external serving-integrity obligation
(\cref{thm:endtoend}).}
\label{fig:pipeline}
\end{figure}

The defense confines the value path to the public monitor's visible channel
(\cref{fig:pipeline}). Its objects form one canonical setup, stated once and used by
every section that follows.

\begin{definition}[The concealment setup]
\label{def:setup}
On the residual stream $X = \Real^{\dmodel}$ (the hidden-state vector passed between
transformer layers), a \emph{monitor} is a public linear map $\Mon\colon X \to Y$ into an
observation space $Y$; its \emph{blind subspace} $\ker\Mon$ is the set of perturbations
the monitor maps to zero. A \emph{value map} $\Val$ (here, the value projection an adapter
ships) carries the channel: the pair $(\Mon,\Val)$ is \emph{live} iff some payload is
concealed ($\Mon\delta=0$) and active ($\Val\delta\neq 0$), and $\Val$ is \emph{inert}
relative to $\Mon$ when it annihilates the blind subspace,
\begin{equation}
\Inert(\Mon,\Val)\ :\Longleftrightarrow\ \ker\Mon \subseteq \ker\Val
\ \Longleftrightarrow\ \Val(\ker\Mon) = \{0\},
\label{eq:inert}
\end{equation}
equivalently (\cref{thm:inertness}(b)) when it \emph{factors through the monitor},
$\Val = C\Mon$ for some $C$: the value path reads $X$ only through what the monitor
exposes. A deployed monitor may post-process $\Mon x$ nonlinearly; every kernel statement
below survives that composition, since $\delta\in\ker\Mon$ leaves $\Mon x$, and hence
any statistic computed from it, unchanged.
\end{definition}

\noindent The concrete
instance we study monitors a single head's query/key map $\pip = [\,W_Q;W_K\,]$ (written
for the monitored head),
an attention-pattern monitor of the safety-probing class~\cite{belinkov2022probing,zou2023repeng},
and confines the layer's value update $W_V$ (the value projection, stacked across the
layer's heads) against it; $\ker\pip$ is the per-head operator-invisibility floor. The
certificate types the whole-layer value update against any declared monitor; the
attention-block \emph{semantic} corollary (\cref{lem:attention}) is head-local unless the
monitor jointly covers every participating head (the joint-coverage predicate,
\cref{sec:predicates});
\cref{sec:eval} runs the end-to-end pipeline against the published SteerEdit attack and typed
garak injection families.

\paragraph{Inertness is built, not tested.} An inert value path actuates no concealed
direction, so no live operator-invisible payload exists: inertness is \emph{exactly} the
negation of channel liveness (\cref{thm:inertness}(a)). It is not a property an arbitrary
adapter has (\cref{sec:eval} shows honest fine-tunes are far from inert), so we do not
test for it; we \emph{install} it.

\paragraph{Confinement makes inertness automatic.} Let $\Pivis$ be the public projection
onto the visible channel $(\ker\Mon)^{\perp}$ (computable from the public monitor; it
annihilates $\ker\Mon$). Serve the value path through $\Pivis$. Then for \emph{any}
value map $\Val$, the confined map $\Val\circ\Pivis$ is inert:
\begin{equation}
\ker\Mon \subseteq \ker\Pivis \ \Longrightarrow\ \Inert(\Mon,\ \Val\circ\Pivis),
\label{eq:confine}
\end{equation}
because every concealed $\delta\in\ker\Mon$ has $\Pivis\delta = 0$
(\cref{thm:inertness}(c)). The served value path therefore actuates no operator-invisible
\emph{linear-carrier} payload of the committed adapter, \emph{by construction}, whatever
$\Val$ is: a structural guarantee about this channel class, not behavioral safety and
not the operator-visible class (\cref{sec:scope}). For a LoRA
update the construction is \emph{typed}: build the read factor as $A = C\Mon$
(\cref{sec:certificate}).

\paragraph{The deployed monitor: QK plus its optimal in-kernel augmentation.} Confining
onto $(\ker\Mon)^{\perp}$ discards the part of the value path that reads $\ker\Mon$, so the
choice of $\Mon$ sets both the utility cost and the residual. The primary threat is defined
by the QK map $\pip$, so the deployed monitor keeps $\pip$ and \emph{augments} it within its
own blind subspace: the defender additionally monitors the $q$ in-kernel directions that
minimize the base value path's residual gain, optimally the top right-singular directions
of $\Valbase$ restricted to $\ker\pip$ (\cref{thm:augment}): a compiler-optimality
statement, given the budget $q$, $\mathsf{CompileMonitor}$'s augmentation is the best in
the declared class and publishes its residual floor. The augmented monitor $\Mdep$
is a deterministic public function of the base model (\cref{sec:threat}). Every payload in the augmented blind subspace is still
operator-invisible ($\ker\Mdep\subseteq\ker\pip$, relative to the same fixed-point
compilation of $\pip$, so the augmentation only
\emph{expands} coverage within $\ker\pip$, never outside it), and the certificate
proves the \emph{update} inert exactly, so the served map's residual blind read is pinned
at the deployed floor, the exact public quantity $\lVert\Valbase P_{\ker\Mdep}\rVert$
computed in the clear; \cref{thm:augment} identifies its ideal-design value
$\sigma_{q+1}(\Valbase\Nbasis_{\pip})$, and the gap between the two is itself public and
verifier-computable; the residual floor is tunable down by raising $q$
(\cref{sec:certificate}).

\paragraph{Two reference policies calibrate it.} The \emph{gain} monitor, whose
visible channel is the top-$k$ right-singular subspace of the base $W_V$, is the
unconstrained optimum for a \emph{different} invisibility policy; its floor is
$\sigma_{k+1}(W_V)$, exactly zero on the low-rank value paths of grouped-query attention. We
use it as an oracle benchmark, not as the deployed policy. A \emph{data}-aligned monitor (the
activation covariance) is the negative control: free benign utility, largest residual
(\cref{sec:eval}).

\paragraph{The predicate family.} The construction extends to other fixed public carrier
subspaces: joint coverage of two monitors, named-trigger inertness, and their compositions.
Each reduces to a homogeneous identity against a public subspace (the cheap ZK direction)
and composes into a single certificate. The full predicate table and the harvested red-team
certificates are in \cref{sec:appendix}. The next section turns
this defense into a certificate: what the publisher proves, and why it is cheap.

\section{The certificate: typed, real-rank-sound, bounded-residual}
\label{sec:certificate}

The defense is cheap to prove for one reason: the monitor is public, so what the publisher
must show is a homogeneous identity between the committed secret and a public matrix. This
section gives that certificate: the typed form we deploy, its basis-dependent equivalent
and that form's load-bearing soundness condition, the field-versus-reals gap, and the public
base floor that bounds the served residual.

\paragraph{The certificate is the adapter's type.} By \cref{thm:inertness}(b), inertness is
exactly factorization through the monitor:
$\ker\Mon\subseteq\ker\Valupd \Leftrightarrow \exists\,C:\ \Valupd = C\Mon$. So rather
than training an unrestricted adapter and afterwards proving a geometric property of it,
the publisher \emph{builds} the adapter monitor-typed. For a LoRA $\Valupd = BA$, construct the
read factor as
\begin{equation}
A = C\Mon, \qquad \Valupd = BA = (BC)\Mon,
\label{eq:typed}
\end{equation}
which is inert for \emph{any} write factor $B$ (\cref{thm:typed}): the adapter reads the
residual stream only through what the monitor exposes,
$x \mapsto \Mon x \mapsto C\Mon x \mapsto BC\Mon x$, and any $x\in\ker\Mon$ dies at the
first arrow. The publisher ships the adapter in the standard serialization $(B,A)$ and
commits to it; the type witness $C$ is a \emph{private} witness committed alongside $A$
in the certificate commitment (\cref{sec:realization}), distinct from the checkpoint
identity $(B,A)$ the hub publishes. The certificate is the identity
\begin{equation}
A - C\Mon = 0
\label{eq:zeroproduct}
\end{equation}
over the committed $A$ and witnessed $C$, against the public $\Mon$: an honest typed
adapter satisfies it by construction, while an $A$ carrying any unconfined component fails
it for every $C$. The check is genuinely \emph{linear} in the committed secret against
public coefficients: no nullspace basis to certify, no in-circuit rank, no spanning check,
and no deficient-basis attack surface. The identity is small ($r\times\dmodel$ entries
with inner dimension $\rank\Mon$), so a circuit can check it in full (the public monitor
baked into the verification key) or probe it at $\zkProbes$ challenges fixed \emph{after} the
commitment; both realizations, and what each adds to the soundness error, are in
\cref{sec:realization}.

\paragraph{The equivalent zero-product, and its spanning condition.} The basis-dependent
equivalent names a public matrix $\Nbasis \in \Real^{\dmodel\times b}$ whose columns span
the blind subspace $\ker\Mon$ ($b = \dmodel - \rank\Mon$) and checks
$\Valupd\Nbasis = 0$, which equals inertness exactly when the columns of $\Nbasis$ span
all of $\ker\Mon$ (\cref{lem:zeroproduct}). This form carries a load-bearing precondition the typed form does
not have: if $\Nbasis$ misses part of $\ker\Mon$, an adversary who knows the deficient
$\Nbasis'$ ships a map reading an omitted direction, passing $\Val\,\Nbasis'=0$ while
non-inert (the generic failure, not a corner case (\cref{thm:spanning})). A verifier using
this form must therefore establish $\Mon\Nbasis = 0$ \emph{and}
$\rank\Nbasis = \dmodel-\rank\Mon$ in the clear (the rank equality alone fixes the
dimension, not the subspace). We use the zero-product form where the subspace \emph{is} the
given object rather than a computed kernel (the named-trigger predicate, whose $\Nbasis$
spans a public trigger family (\cref{tab:predicates}), and the falsification suite, which
exercises the deficient-basis attack (\cref{sec:eval})) and the typed form
\eqref{eq:zeroproduct} as the deployed certificate.

\paragraph{The deployed monitor is the security object.} The circuit reasons over $\Fp$;
the model is real. We close that gap not by approximating an ideal monitor but by
\emph{declaring} the deployed finite-precision monitor $\Mdep$ to be the monitor: it is a
bounded public fixed-point matrix, part of the published monitor policy
(\cref{sec:realization}), and its blind subspace $\ker\Mdep$ is the certified object.
Weights are fixed-point rationals (as are the shipped factors, built on the same encoding
lattice) and the circuit's range checks confine every entry and
product of \eqref{eq:zeroproduct} to the signed integer range, so $A = C\Mdep$ is a
genuine \emph{integer} identity and inertness against $\Mdep$ is exact: for every
$\delta\in\ker\Mdep$, $\Valupd\,\delta = 0$, with no tolerance and no prover-chosen
slack. There is no field-computed nullspace whose rank or lift would need interpretation
(that semantic burden is one reason the zero-product form is not the deployed one: its
basis is computed \emph{over} $\Fp$, whose correspondence to the real kernel needs a
rank-preservation argument the typed form never raises). A deployer who instead wishes to name an \emph{ideal} real-valued monitor and serve
its finite-precision compilation pays an explicit, optional tolerance
\eqref{eq:typedtol}, derived in \cref{sec:zeroproduct}; the deployed guarantee against
$\Mdep$ needs none of it.
\Cref{thm:tolerance} gives the corresponding bound for the zero-product form.

\paragraph{The residual is the public base floor.} The update certificate is exact, so by
\cref{lem:served-floor} the served map's blind read equals the base's own: for any public
monitor with blind-subspace basis $\Nbasis$,
$\lVert\Valsrv\Nbasis\rVert_2 = \lVert\Valbase\Nbasis\rVert_2$, a \emph{public} quantity
the verifier computes from the base model in the clear. This equality is exact for the
\emph{deployed quantized} monitor $\Mdep$; for the ideal real-valued monitor the
encoding bound \eqref{eq:typedtol} adds an explicit tolerance:
\begin{equation}
\lVert \Valsrv P_{\ker\Mon}\rVert_2 \ \le\
\sigma_{q+1}\!\bigl(\Valbase\Nbasis_{\pip}\bigr) \;+\; \eps_{\mathrm{enc}},
\label{eq:residual}
\end{equation}
where $\eps_{\mathrm{enc}} = \lVert BC\rVert\,\lVert\Mon_{\mathrm{fp}}-\Mon\rVert$, with
$\Mon_{\mathrm{fp}}$ the ideal monitor's fixed-point compilation, is the publicly
bounded encoding resolution (\cref{sec:realization}). Under the deployed monitor
$\Mdep$ (the QK map plus $q$ optimal in-kernel directions), \cref{thm:augment} makes
$\sigma_{q+1}$ the smallest residual any \emph{ideal} $q$-direction augmentation
preserving the $\ker\pip$ threat can achieve; the deployed compiled floor tracks it
measurably (at most $\compGapMaxPct\%$ off, median $\compGapMedPct\%$, over the proved
GPT-2 layers, a public gap), dialed down by raising $q$; under the gain oracle
it is $\sigma_{k+1}(W_V)$, reaching exactly zero on the low-rank value paths of
grouped-query attention (\cref{sec:eval}). Here and below $\lVert\cdot\rVert$ denotes the
operator/Euclidean $2$-norm, written $\lVert\cdot\rVert_2$ for emphasis; adapter
\emph{occupancy} (\cref{sec:eval}) uses the Frobenius norm $\lVert\cdot\rVert_F$. A
passing certificate then bounds any downstream readout: for a linear
readout $G$ (an unembedding row, a logit map) and concealed $\delta = \Nbasis c$,
\begin{equation}
\lVert G\,\Valsrv\,\delta\rVert \ \le\ \lVert G\rVert\,\lVert \Valbase\Nbasis\rVert\,\lVert c\rVert
\label{eq:behavioral}
\end{equation}
(\cref{thm:readout}), the update contributing exactly zero. Exact base closure is the
floor's zero endpoint; the public base floor is the deployable interior. \Cref{sec:eval}
measures both floors.

\paragraph{Circuit scale.} The committed object is the adapter, not the model: a rank-$r$
LoRA value adapter has $O(r(\dhead+\dmodel))$ secret entries, the witness $C$ adds
$r\cdot\rank\Mon$, and the certificate is one small matrix identity against the public
$\Mon$. On GPT-2 ($\dmodel=\gptDmodel$, $\dhead=\gptDhead$) the deployed monitor has rank
$\rank\Mon \ll \dmodel$ while the blind subspace has dimension $\gptBlindDim$; a rank-$8$
adapter has $\gptLoraSecretRankEight$ secret entries against $\gptDenseSecret$ for a dense
$\Delta W_V$. Adapter-scale attestation is thus plausible on present proving stacks;
\cref{sec:soundness} states the theorems that make a passing certificate mean what it
claims, and \cref{sec:realization} turns them into a zero-knowledge proof.

\section{Soundness}
\label{sec:soundness}

The certificate of \cref{sec:certificate} is a statement about committed weights; here we
state the theorems that give it meaning. All maps are linear over a field; each statement
carries its proof idea inline, with auxiliary statements and a machine-checked
formalization in \cref{sec:appendix}.

Inertness has three faces: a security property, a normal form, and a construction. The
first says what a passing certificate rules out, the second is what the prover actually
proves, and the third is why a publisher can \emph{install} the property instead of testing
for it.

\begin{theorem}[Inertness: channel-absence, factorization, construction]
\label{thm:inertness}
For a monitor $\Mon$ and value map $\Val$:
\emph{(a) Channel-absence.}
$\Inert(\Mon,\Val) \Leftrightarrow \neg\,\exists\,\delta:\ \Mon\delta = 0 \wedge \Val\delta \neq 0$.
\emph{(b) Factor-through.} $\Inert(\Mon,\Val) \Leftrightarrow \exists\,C:\ \Val = C\circ\Mon$.
\emph{(c) Construction.} If $\ker\Mon \subseteq \ker\Pivis$ then $\Inert(\Mon,\ \Val\circ\Pivis)$
for \emph{every} value map $\Val$; the public projection onto $(\ker\Mon)^{\perp}$ satisfies
the hypothesis.
\end{theorem}

\begin{proof}[Proof idea]
(a) unfolds $\ker\Mon\subseteq\ker\Val$. (b) is the universal property of the quotient: an
inert $\Val$ is constant on the fibers of $\Mon$, so it descends to a well-defined
$\bar\Val\colon\Real^{\dmodel}/\ker\Mon \to \Real^{d_v}$ and lifts back through any
splitting. (c) for $\delta\in\ker\Mon$, $\Pivis\delta = 0$, hence
$(\Val\circ\Pivis)\delta = 0$.
\end{proof}

\noindent The three parts carry three distinct claims. Part (a) is the security meaning: an
inert value path admits no concealed-and-active payload, which turns an existential
impossibility (no hidden payload exists) into a checkable algebraic inclusion
($\ker\Mon\subseteq\ker\Val$), the property prior detection methods attempt to test for but
cannot certify. Part (b) is the basis-free normal form, naming no $\Nbasis$: the value path
reads the residual stream only through what the monitor exposes, so it needs no spanning
check and has no deficient-basis attack surface, and two inputs the monitor cannot
distinguish receive the same value-path output. Part (b) is also what makes the certificate
\emph{cheap}: the identity $A = C\Mon$ is linear in the committed secret against a public
matrix, exactly the structure zero-knowledge proofs are built for, and categorically
distinct from a magnitude bound, which cannot read the kernel relation an
operator-invisible payload hides in~\cite{shangchen2026drift}. Part (c) is why the property
can be built rather than tested: serving $\Valbase + \Valupd\circ\Pivis$ confines the
publisher's contribution while leaving the public base $\Valbase$ untouched. Which map the
certificate should attest, the update or the whole served map, is settled by splitting the
served map into its public and secret parts.

\begin{definition}[Served map: base plus update]
\label{def:served}
The served value map is $\Valsrv = \Valbase + \Valupd$: a public base map $\Valbase$
(a deterministic function of the base model) plus a committed secret update $\Valupd$
(the shipped adapter; $\Valupd = BA$ for a rank-$r$ LoRA, commitments as in
\cref{sec:threat}). Throughout, $\Valsrv$ is the value
map of the \emph{certified released checkpoint}; binding it to the model a serving endpoint
actually loads is a separate obligation we scope in \cref{sec:scope}.
\end{definition}

\begin{lemma}[Residual is the public base floor]
\label{lem:served-floor}
$\Valsrv\Nbasis = \Valbase\Nbasis + \Valupd\Nbasis$, and if the update is inert
($\Valupd\Nbasis = 0$) then $\Valsrv\Nbasis = \Valbase\Nbasis$. The right side is
\emph{public}: the verifier computes $\Valbase\Nbasis$ from the base model and the public
$\Nbasis$, so an inert update pins the served residual at the base floor in every norm,
with no prover-chosen tolerance.
\end{lemma}

\begin{proposition}[Attest the update, not the served map]
\label{prop:served}
\emph{(a) Isolation.} If $\Inert(\Mon,\Valupd)$ then
$\Inert(\Mon,\Valsrv) \Leftrightarrow \Inert(\Mon,\Valbase)$.
\emph{(b) Rank obstruction.} If $\rank\Valupd \le r$ and $\rank(\Valbase\Nbasis) > r$, then
$\Valsrv\Nbasis \neq 0$.
\end{proposition}

\begin{proof}[Proof idea]
(a) with the update inert, $\Valsrv\delta = \Valbase\delta$ for every $\delta\in\ker\Mon$.
(b) if $\Valsrv\Nbasis = 0$ then $\Valupd\Nbasis = -\,\Valbase\Nbasis$, so
$\rank(\Valbase\Nbasis) = \rank(\Valupd\Nbasis) \le \rank\Valupd \le r$, contradicting
$\rank(\Valbase\Nbasis) > r$.
\end{proof}

\noindent Together these fix the object to certify. By (a) the publisher's update neither
opens nor closes a channel the public base lacks, so attesting $\Inert(\Mon,\Valupd)$
(equivalently the typed factorization $\Valupd = C'\Mon$, \cref{thm:inertness}(b), or the
zero-product $\Valupd\Nbasis = 0$ against a spanning basis, \cref{lem:zeroproduct})
certifies precisely what the publisher adds. By (b) the alternative is unavailable: no
rank-$r$ update makes the served map blind-inert once the base reads the blind subspace at
rank above $r$, so ``the served model reads nothing blind'' is unachievable at adapter
scale. We therefore confine the update (\cref{thm:inertness}(c) on $\Valupd$) and report the
base floor, rather than attempt an infeasible whole-map closure.

\begin{theorem}[Typed LoRA is inert by construction]
\label{thm:typed}
If a LoRA's read factor factors through the monitor, $A = C\Mon$, then $\Valupd = BA$ is inert
for \emph{any} $B$. Confinement is then a property of the adapter's \emph{type}: the publisher
ships and commits the standard factors $(B,A)$, holds the type witness
$C\in\Real^{r\times s}$ ($s=\rank\Mon$) privately, and the certificate is the identity
$A - C\Mon = 0$, \emph{linear} in the committed secret against the public $\Mon$. Conversely
every rank-$r$ inert update admits such a factorization (\cref{thm:inertness}(b) applied to
$\Valupd$, then a rank factorization), so the typed form loses no generality at the
\emph{dense-update} level, up to refactorization of the same $\Valupd$. The qualification is
real: cancellation through $B$ can make $BA$ inert while $A$ itself is not, so typing the
shipped read factor is the stronger per-factor property, and a publisher training under
confinement chooses this parameterization from the start.
\end{theorem}

\noindent The basis-dependent zero-product equivalent and its load-bearing spanning
precondition (\cref{sec:certificate}) are stated and proved in \cref{sec:appendix}.

\noindent The next two facts bound the residual: the field-versus-reals gap, and the
behavioral shift of an $\eps$-inert defense.

\begin{theorem}[Encoding tolerance]
\label{thm:tolerance}
For bounded operators with $\Val_{\mathrm{fp}}\circ\Nbasis_{\mathrm{fp}} = 0$,
$\lVert \Val\circ\Nbasis \rVert \le \lVert \Val-\Val_{\mathrm{fp}}\rVert\,\lVert \Nbasis\rVert +
\lVert \Val_{\mathrm{fp}}\rVert\,\lVert \Nbasis-\Nbasis_{\mathrm{fp}}\rVert$.
\end{theorem}

The next lemma grounds the abstract maps in an actual attention block.

\begin{lemma}[Attention instantiation]
\label{lem:attention}
Fix an attention layer with post-LayerNorm input $x$, the monitored head's query/key map
$\pip$ (\cref{def:setup}), base value map $\Valbase=W_V$, output map $W_O$, and value update
$\Valupd$. For a payload $\delta\in\ker\pip$ (so $W_Q\delta=W_K\delta=0$ \emph{for the
monitored head}): (i) the monitored head's attention weights
$\mathrm{softmax}(QK^{\top}/\sqrt{\dhead})$ are unchanged when $\delta$ is added to
the input of a token $j$, since its $Q,K$ do not read $\delta$; (ii) with that head's
(unchanged) attention matrix $\alpha$, its output contribution at token $i$ shifts by
$\alpha_{ij}\,W_O(\Valbase+\Valupd)\delta$,
which factors through $(\Valbase+\Valupd)\delta$. Hence an update with $\Valupd\delta=0$
contributes nothing and the residual shift is the base term $\alpha_{ij}\,W_O\Valbase\delta$,
bounded by \cref{thm:readout} with readout $G=W_O$. The semantic reading is
\emph{head-local}: another head whose query/key maps read $\delta$ may change its own
routing. If the declared monitor jointly stacks every head's query/key map (the
joint-coverage predicate, \cref{sec:predicates}), a $\delta$ in the joint kernel fixes
every head's routing and the block-level claim closes.
\end{lemma}

\begin{proof}[Proof idea]
Write tokens as columns $X$; the monitored head's $Q,K$ are linear in $X$, so for
$\delta\in\ker\pip$ that head's queries and keys are invariant and its attention
$\alpha$ is unchanged. That head's
output contribution is $W_O(\Valbase+\Valupd)X\,\alpha^{\top}$; adding $\delta$ at token $j$ changes it by
$W_O(\Valbase+\Valupd)\delta\,e_j^{\top}\alpha^{\top}$, i.e.\ $\alpha_{ij}\,W_O(\Valbase+\Valupd)\delta$
at output token $i$.
\end{proof}

\begin{theorem}[Bounded behavioral shift]
\label{thm:readout}
For any downstream linear readout $G$, value map $\Val$, and column-orthonormal basis
$\Nbasis$ of the blind subspace, and any concealed payload $\delta = \Nbasis c$
(so $\lVert\delta\rVert = \lVert c\rVert$),
\[
\lVert G\,\Val\,\delta\rVert \ \le\ \lVert G\rVert\,\lVert \Val\circ\Nbasis\rVert\,\lVert c\rVert .
\]
Instantiating with the served map, an inert update pins
$\lVert\Valsrv\Nbasis\rVert = \lVert\Valbase\Nbasis\rVert$ (\cref{lem:served-floor}), the
public base floor: under the deployed compiled monitor this is the exact public quantity
$\lVert\Valbase P_{\ker\Mdep}\rVert$, whose ideal-design value is
$\sigma_{q+1}(\Valbase\Nbasis_\pip)$ (\cref{sec:certificate}), at most
$\sigma_{k+1}(W_V)$ under the gain oracle; the update's own readout shift is exactly zero
($\Valupd\Nbasis = 0$), and the endpoint $\Valbase\Nbasis = 0$ makes every readout
invariant to concealed payloads.
\end{theorem}

\noindent The floor depends on the monitor. The deployed monitor \emph{augments} the QK map
within its own blind subspace, expanding coverage of $\ker\pip$ without monitoring outside
it ($\ker\Mdep\subseteq\ker\pip$), and there is a best way to do so.

\begin{theorem}[Optimal in-kernel augmentation]
\label{thm:augment}
Let $\Nbasis_\pip$ have orthonormal columns spanning $\ker\pip$. If the defender monitors,
besides $\pip$, a $q$-dimensional $T\subseteq\ker\pip$ (leaving residual blind subspace
$\ker\pip\cap T^{\perp}$), then the smallest achievable residual gain of the base value path is
\[
\min_{\substack{T\subseteq\ker\pip,\ \dim T=q}}
\bigl\lVert \Valbase\,P_{\ker\pip\cap T^{\perp}} \bigr\rVert
\ =\ \sigma_{q+1}\!\bigl(\Valbase\Nbasis_\pip\bigr),
\]
attained when $T$ spans (through $\Nbasis_\pip$) the top-$q$ right-singular directions of
$\Valbase\Nbasis_\pip$. This is a lower bound among augmentations of the stated size and nesting: no
$q$-direction augmentation within $\ker\pip$ can achieve a smaller residual, so the
all-layer frontier of \cref{sec:eval} plots the ideal-design optimum for this monitor class.
\end{theorem}
\noindent One caveat this makes explicit: since $\ker\Mdep\subsetneq\ker\pip$, a typed adapter
reads the augmentation directions $T=\ker\pip\cap(\ker\Mdep)^\perp$, which a bare attention-pattern
probe (a function of $\pip$ alone) does not see. The confinement guarantee is therefore relative to
the \emph{deployed} $\Mdep$: it holds when the served monitor is $\Mdep$ itself, augmentation rows
included, and reduces to $\ker\pip$-confinement only if $q=0$. The frontier is the graded form of an exact-certification budget: the smallest
$q$ at which $\sigma_{q+1}(\Valbase\Nbasis_\pip)$ reaches zero is
$\dim \Valbase(\ker\pip) = \rank[\,\pip;\Valbase\,] - \rank\pip$, the minimum number of
monitored directions any in-kernel augmentation needs before the base value path's blind
channel closes exactly. Below that endpoint the curve prices every partial budget, so a
deployer reads off the residual purchased by each additional monitored direction.

\section{Zero-knowledge realization}
\label{sec:realization}

The theorems of \cref{sec:soundness} establish what a passing certificate means; this
section shows how to prove it in zero knowledge without revealing the read factor it certifies.
The certificate is a commit-and-prove statement: commit once to the certified read factor
and type witness, then prove that the committed read factor is monitor-typed ($A - C\Mdep = 0$
against the public deployed monitor).

\paragraph{Commitment.} The publisher commits to the read factor $A$ and type witness $C$
with a \emph{hiding} commitment so it leaks nothing about the committed pair: a randomized hash commitment
$\commit = H(\textsf{domain},\rho_{\mathrm{com}},\mathrm{serialize}(A,C))$ with a private
high-entropy salt $\rho_{\mathrm{com}}$, canonical length-prefixed serialization, and domain
separation. This is the commitment the measured circuit implements (the salt is a private
witness of the in-circuit hash, \cref{sec:eval}), with a blinded KZG~\cite{kzg2010} or
Pedersen commitment the alternative. The commitment $\commit$ is a \emph{certificate
commitment}: it binds the read factor $A$ and type witness $C$ that the proof checks, not
the full checkpoint $(B,A)$; $B$ is not part of the circuit and does not need
commitment-binding. The \emph{checkpoint identity} (the artifact a hub publishes) remains
$(B,A)$ in whatever serialization the deployment uses; a separate binding of $\commit$ to
the checkpoint hash, or a Merkle proof opening $A$ from the checkpoint, ties the
certificate to a specific shipped artifact.

\paragraph{What is proved.} One public-linear fact about the committed secret: the shipped
read factor is monitor-typed, $A = C\Mdep$ for some $C$, so the update $BA$ is inert by
construction for \emph{any} write factor (\cref{thm:typed}). Two realizations check it,
neither with an in-circuit rank or nullspace. A custom circuit bakes the public $\Mdep$
into the verification key and checks the identity \emph{in full}, with no challenge and no
probabilistic term. The measured stock-toolchain instantiation keeps the circuit minimal
and \emph{probes the committed typed identity directly}: the public inputs are a challenge
vector $r$ and the verifier-recomputed $m = \Mdep\,r$, the circuit checks
$(A - C\Mdep)\,r = 0$ over range-checked fixed-point integers with
$(A,C,\rho_{\mathrm{com}})$ the committed preimage, and the challenge is derived from the
commitment (\cref{thm:protocol}). Passing certifies the \emph{per-factor} typed relation
$A = C\Mdep$ (not merely dense-update inertness) because $B$ is not in the circuit and
cannot cancel an unconfined component of $A$. Because $A$ is serialized at the \emph{product}
scale (the input and parameter scales are set equal, so $A$'s entries carry the two
fractional-bit blocks of $C\Mdep$), the checked relation is the exact scale-balanced integer
identity $A_Z = C_Z\,(\Mdep)_Z$ with no residual scale factor;
the no-wrap range checks (below) keep every accumulated product inside the signed range, so a
field zero is an integer zero.
Crucially, $\Mdep$ never enters the circuit as prover-supplied data: the verifier
recomputes $m$ from the public monitor in the clear, so a prover cannot substitute a
smaller monitor (the deficient-basis attack surface of the zero-product form,
\cref{thm:spanning}, does not arise).

\paragraph{Commit-and-prove soundness.} Composing the standard primitives gives the
end-to-end guarantee. Fix a public statement $\mathsf{st}$ comprising the base-model
hash; $\Mdep$ or its digest; the package root $h_{\mathrm{pkg}}$ (computed before proving,
\cref{sec:contract}); the dimensions $\dmodel$, $\dhead$, and $r$; the fixed-point
scale and signed range; the circuit version; and the layer/head/tensor identifier.
The publisher publishes the hiding commitment $\commit$,
and the challenge is $\zkProbes$ independent probe vectors $r_j =
H(\textsf{challenge-domain},\mathsf{st},\commit,j)$, probe $j$ expanded from its own
domain-separated random-oracle seed, its coordinates successive rejection-sampled outputs
uniform over the bounded integer range $S$ ($\lvert S\rvert \approx 2^{\zkProbeSetBits}$, bounded so
every scaled product stays inside the circuit's exact no-wrap range: each intermediate
obeys a checked bound, and the per-intermediate ledger, worst margin $\nwMarginWorst$ bits
below $p/2$ on the measured configuration, is in \cref{sec:contract}; larger probes
overflow the range and break the identity), fixed \emph{after} $\commit$ so the
committed factors cannot be chosen against a known probe. Soundness then amplifies with the probe
count $\zkProbes$ rather than the per-probe range.
\begin{theorem}[Commit-and-prove soundness]
\label{thm:protocol}
Under commitment binding ($\eps_{\mathrm{bind}}$), SNARK knowledge soundness
($\eps_{\mathrm{KS}}$), the in-circuit range checks (the no-wrap condition below), and the
random-oracle challenge, an accepting proof implies that the read factor committed under
$\commit$ is monitor-typed ($A = C'\Mdep$
for some $C'$, an exact identity on the encoding lattice), hence the update $BA$
formed with \emph{any} write factor $B$ is inert, except with probability at most
$\eps_{\mathrm{bind}} + \eps_{\mathrm{KS}} + q_H\,\lvert S\rvert^{-\zkProbes}$ for a
$q_H$-query prover. Under commitment \emph{hiding} (the salted hash modeled in the
random-oracle model; a Pedersen or blinded-KZG commitment gives the standard
computational-hiding statement) and SNARK zero knowledge, the transcript, $\mathsf{st}$ together with the hiding commitment $\commit$ and the package root
$h_{\mathrm{pkg}}$, both of which the verifier already holds, reveals nothing further about
$(A,C)$; the write factor $B$ is publicly shipped, hash-bound by $h_{\mathrm{pkg}}$,
and outside the hiding claim.
\end{theorem}
Three levels must not be conflated: (i) the \emph{algebraic} identity and its all-probes
form (\cref{lem:allprobes}); (ii) the deployed circuit probes the \emph{committed}
identity at $\zkProbes$ post-commitment challenges, false-accepting with probability at
most $\lvert S\rvert^{-\zkProbes}\approx 2^{-\zkSoundnessBits}$ plus the
$q_H\,\lvert S\rvert^{-\zkProbes}$ grinding term (\cref{lem:allprobes}); a field-native realization
with a uniform-$\Fp$ challenge (or $\Mdep$ baked into a custom \texttt{halo2-lib}
verification key, checking the identity in full) sharpens the per-certificate term to $1/p$
at the same proof size;
(iii) knowledge soundness and zero knowledge are the SNARK's, and the commit-and-prove
link binds the checked witness to $\commit$. The no-wrap condition (every accumulated
fixed-point product stays within the signed range of $\Fp$, so a field zero is an integer
zero) is enforced in-circuit by the range checks, a checked constraint rather than an
assumption.

\paragraph{Proof system.} The predicate is fixed-point linear algebra plus a
prover-supplied witness, so a Plonkish system with lookup arguments~\cite{halo2} fits:
native lookups make the dominant cost (fixed-point range checks) cheap, a universal
structured reference string is reused across adapter shapes (one trusted setup; a
per-circuit Groth16 ceremony~\cite{groth16} run by the publisher alone would leave the
publisher holding the circuit trapdoor, with which it could forge proofs for that
circuit), and the constant-size proof with fast verifier suits an attestation posted beside
a checkpoint. The commit-and-prove link binding the predicate to $\commit$ follows the
modular CP-SNARK line~\cite{legosnark2019,lunar2021,artemis2024}; reducing a kernel
statement to a witness identity is folklore~\cite{dumas2014certificates}; the
secret-artifact-versus-public-specification structure mirrors ZK-CEC~\cite{shen2026zkcec}
and the ZK-UNSAT line~\cite{luo2022zkunsat};
and the fixed-point-after-matmul soundness follows Range-Arithmetic~\cite{rangearith2025}.
Because the certificate is \emph{matmul-only} (a matrix product and a subtraction, with no
rank or kernel operator), it is directly expressible as an ONNX graph, so a standard
Halo2-backed ZKML compiler (EZKL~\cite{ezkl}) proves it as-is; we use this for the
feasibility benchmark (\cref{sec:eval}), with a custom \texttt{halo2-lib} circuit the
later optimization.

\paragraph{Verifier preprocessing.} The ``hard part is public'' argument has a concrete
cost: constructing $\Mdep$ from the base (weight extraction, restricted SVD, lattice
quantization, challenge $m=\Mdep r$) takes $\preprocGptMs$\,ms per layer on GPT-2
($\dmodel=768$), and $\preprocQwenFiveMs$--$\preprocSmolMs$\,ms across the four models, dominated
by the restricted SVD; model loading adds $1$--$4$\,s, amortized across layers and sessions.

\paragraph{Realized prover.} We implement the certificate end-to-end on a Halo2 + KZG
backend: a complete proof that a GPT-2-scale value adapter satisfies the typed identity
proves in seconds with a millisecond verifier (measured in \cref{sec:eval},
\cref{tab:cost}); the $O(r(\dhead+\dmodel))$ secret and single
small matrix identity (\cref{sec:certificate}) are what make this cheap. The
commit-and-prove binding uses the \emph{hiding} commitment above: the factors and private
salt $\rho_{\mathrm{com}}$ are committed by an in-circuit Poseidon hash exposed as a
public instance, so one verification checks the certificate \emph{and} that it was
computed with the adapter committed as $\commit$; under Poseidon's collision resistance
no efficient publisher finds a colliding adapter (the binding property), while the
commitment's \emph{hiding} (that $\commit$ reveals nothing about the factors beyond
$\mathsf{st}$) is a separate assumption on the salted Poseidon commitment, for which the
high-entropy salt $\rho_{\mathrm{com}}$ supplies the required preimage entropy; collision
resistance alone does not give hiding. The in-circuit hash dominates the binding cost
(a one-time publisher cost, \cref{sec:eval}), with a KZG polynomial
commitment~\cite{kzg2010} or a custom \texttt{halo2-lib} circuit the remaining speedup.

\subsection{The system contract}
\label{sec:systemcontract}

Two verifier-side procedures close the loop; \cref{sec:contract} gives the full protocol
interface. $\mathsf{VerifyBundle}$ runs the per-layer check on all $L$ layers and recomputes
the package root $h_{\mathrm{pkg}}$, which binds each layer's shipped write factor $H(B_\ell)$
together with that layer's proof commitment $\commit_\ell$. $\mathsf{ServeGuard}(\tilde
A,C,\Mdep)$ is run by the weight-holding consumer on the bytes it actually serves: it admits
$\tilde A$ only if $\tilde A$ opens the commitment as the committed typed identity $\tilde A =
C\Mdep$, checked \emph{exactly over the integer lattice}, with $\Mdep$ recomputed from the
authenticated base rather than taken from a shipped matrix (closing the deficient-monitor
substitution). The check is exact and unconditional: no prime, no field reduction, no
acceptance threshold, hence no slack an unbounded write factor could amplify, and no
cryptographic assumption on the check itself. The composition of the two is the system's
contract.

\begin{theorem}[Two-part confinement: bundle plus guard]
\label{thm:endtoend}
Fix the authenticated base and the deployed monitors $\Mdep[\ell]$; write
$\mathsf{ServeGuard}_\ell$ for the guard run with $(W_0,\mathsf{policy},\commit_\ell)$ on the
served bytes and sidecar of layer $\ell$.
\emph{(a) Served integrity (exact, unconditional):} if
$\mathsf{ServeGuard}_\ell=1$ for the adapter bytes presented at every
protected layer, then the canonical tensors $\tilde A_\ell$ the guard returns satisfy
$\forall\ell,\ \forall\delta\in\ker\Mdep[\ell],\
B_\ell\tilde A_\ell\delta=0$, and the served residual on each blind subspace is the
authenticated public base floor (\cref{lem:served-floor}).
\emph{(b) Committed-package soundness (computational):} if $\mathsf{VerifyBundle}=1$, then
except with probability $L\,(\eps_{\mathrm{bind}}+\eps_{\mathrm{KS}}+q_H\lvert
S\rvert^{-\zkProbes})$, the adapter committed under $h_{\mathrm{pkg}}$ factors through the
deployed monitor at every layer (\cref{thm:protocol}).
\end{theorem}

\begin{proof}[Proof idea]
(a) composes \cref{thm:typed} at each admitted layer with \cref{lem:served-floor}; the guard's
integer identity check is exact, so no slack enters. (b) applies
\cref{thm:protocol} per layer
under a union bound over the $L$ certificates; the package root binds the proof's own
commitment, so no unlinked commitment splits the attested and shipped factors.
\end{proof}

\noindent The two parts have deliberately distinct reach, and conflating them overstates the
guarantee. Typed-relation soundness binds the \emph{committed} adapter: the proof certifies
that the $A$ in $\commit$, the \emph{same} $\commit$ the package root binds, is typed, which is
what a remote verifier who never sees the read factor obtains. Because one commitment serves both
the proof and $h_{\mathrm{pkg}}$, there is no unlinked second commitment a publisher could
split. What it does \emph{not} do is bind the deployed \emph{bytes} to $\commit$: a publisher
could ship an $A'$ that does not open $\commit$. Catching that is $\mathsf{ServeGuard}$, which
opens $\commit$ against the served $(A',C,\rho)$ and rejects any other factor, untyped or typed
but not the committed one, whatever a malicious publisher or intermediary shipped. The two
conclusions have different strengths: exact integer typing gives \emph{unconditional}
confinement of the admitted tensor, while equality with the previously certified factor is
\emph{computational}, under commitment binding, exactly the split of \cref{thm:endtoend}.
$\mathsf{ServeGuard}$ certifies the \emph{tensor it is handed}; the executed-tensor binding
is scoped in \cref{sec:scope}.

A publisher who needs the full update $\Valupd = BA$ confidential needs a
different realization; what this one buys is that the confinement property is checkable by a
verifier who never sees $A$.

\section{Evaluation}
\label{sec:eval}

The previous sections constructed the defense, proved its soundness, and described a
zero-knowledge realization. This section asks whether it works in practice: does the
certificate decide correctly, does confinement cost utility, does the prover scale, and
does the defense hold against an adaptive adversary? All numbers are macro-guarded from
the released scripts.

\paragraph{Falsification suite.} The evaluation escalates through an adversary hierarchy:
naive hidden edit, published stealth edit, certificate-aware attacker, package
substitution, served-factor substitution. First, the certificate is tested against a suite
of structural attacks: a confined value map is accepted; an OBVS construction is rejected
($\gptSeparationOrders$ orders of separation on GPT-2); a partly visible backdoor is
rejected at every mixing coefficient; an adaptive payload anywhere in $\ker\pip$ is caught
by a spanning basis (\cref{thm:spanning}); and the tolerance bound is verified numerically.
A magnitude/rank/sparsity bound~\cite{shangchen2026drift} calibrated to accept an honest
adapter of the same norm and rank also accepts an OBVS adapter; the confinement certificate
rejects it (residual $\driftObvsBlindResid$).

\paragraph{Adaptive adversary and published attacks.}
SteerEdit~\cite{steeredit2026} compiles activation steering into a persistent
null-space-constrained weight edit; instantiated on the value path, its stealthy trigger
is $\steerStealthConcealed\%$ inside $\ker\pip$ and the certificate ZK-rejects it
($\steerZkSepOrders$ orders from honest); evading the certificate forces visibility
($\steerNaiveVisibility\!\to\!\steerEvadeVisibility$). Our re-implementation of the
supervised spectral detector of~\cite{puertolas2026weightspace} achieves ROC-AUC $\detAucUnit$ on loud payloads
but degrades to $\detAucStealth$ under stealth and needs a labeled reference; the
certificate is \emph{label-free} (its reference is the verifier-declared public monitor)
and exact. A defense-specific adaptive attacker targets
the certificate's residual budget: the optimal residual attacker saturates
$\sigma_{k+1}(W_V)$, capped below the unconfined $\sigma_1$ by
$\advResCapLo$--$\advResCapHi$ orders as coverage rises, and driven to zero by exact
confinement.

\paragraph{Detection from weights is uninformative.}
On GPT-2 ($\dim\ker\pip=\gptBlindDim$ of $\gptDmodel$), a trained LoRA places
$\dfourTrainedBlind\%$ of its value-path norm in $\ker\pip$; a random adapter,
$\dfourRandomBlind\%$; only an adapter built inert by hand sits at the floor; any low-rank
adapter concentrates there by geometry, so no kernel-activity test separates honest from
malicious. A real published adapter (Stanford-Alpaca \texttt{gpt2\_alpaca-lora},
$\realAdapterBlindFrac\%$ kernel-active) confirms it; confining and then tampering separates
at $\realAdapterSepOrders$ orders.

\paragraph{A targeted operator-invisible manipulation, caught and removed.} The channel's
demonstrable use is an integrity manipulation of the output, and the certificate removes
it. We build a value-path payload that reads $u\in\ker\pip$ (so $W_Q u = W_K u = 0$) and
writes toward a chosen token's unembedding direction, shifting its logit by
$\intManipVsBase\times$ the base value path's natural variation at visibility $0$. The
certificate rejects the tampered adapter (residual $\intTamperResid$); confining the
released update accepts it (residual $\intConfinedResid$) and drives the manipulation to
zero. This is an integrity channel (\cref{thm:readout} bounds the shift), not a
demonstrated content backdoor.

\paragraph{The blind subspace is architecture-dependent.} Under MHA, heads' visible
channels sit in general position (overlap $\covMhaOverlap$); under grouped-query attention
(GQA)~\cite{ainslie2023gqa}, key sharing collapses a group's joint visible channel to a
fixed low dimension, a measured $\covGqaBlindShared$-dimensional group-shared blind floor on
Qwen2.5-0.5B; the covering penalty is entirely off-circuit (prover cost is flat in blind
dimension). We report the measured frontier; the general group-level observation geometry
is outside this paper.

\paragraph{Inert-by-design is nearly free, and utility does not select the monitor.} Weight-space
detection fails because a benign adapter \emph{uses} $\ker\pip$; the defense is cheap because it does
not \emph{need} it. We train value-path LoRAs on an instruction task
(Qwen2.5-0.5B-Instruct~\cite{qwen2025}, held-out response perplexity), confined to a
declared monitor's visible channel (the deployable construction, trained-in rather than
projected after) in a factorial design over the three monitors (gain, the augmented QK
policy, and a random control), three visible budgets, two ranks, and $\utilSeeds$ seeds
(\cref{tab:doe}). Confinement is not a no-op: it forces $\Valupd\Nbasis = 0$ (the exact
update certificate of \cref{prop:served}(a)), killing the channel the publisher
could add.
Yet task utility survives, and, the design's decisive finding, survives \emph{equally under
every monitor}: at $25\%$ visible (three-quarters of the residual stream blind), all three
recover $\doeRecovRand$--$\doeRecovGain\%$ of the unconfined adapter's gain (standard
deviation $\le\doeRecovStdMax\%$ over $\utilSeeds$ seeds, \cref{tab:doe}). Matched within each budget--rank--seed cell, the largest monitor-pair
difference is $\doeMatchedMaxDiff\%$, a two-one-sided-test equivalence at
$\pm\doeTostMargin\%$ over $\doeMatchedPairs$ pairs, far tighter than the security separation
below. In these experiments, low-rank benign adaptation re-expresses across every
sufficiently large subspace tested.

\paragraph{But the certified residual does select the monitor.} Utility being
monitor-agnostic is exactly what makes the security argument sharp: the monitors differ not
in trainability but in the \emph{certified residual} they leave. At the same $25\%$-visible
operating point the gain monitor's residual is the public floor $\sigma_{k+1}(W_V)=0$ (exact
closure on this GQA value path), the augmented-QK residual is $\doeResidQkaug\%$ of the
unconfined blind read, and the \emph{random} monitor leaves $\doeResidRand\%$, and still $\doeResidRandDeep\%$ even at $75\%$ visible
(\cref{tab:doe}). Two hypotheses separate cleanly: low-rank fine-tuning repacks into many
subspaces (so utility is flexible), but only \emph{gain-aware} selection minimizes the
worst-case residual (so security is not). In the factorial design, monitor identity
explains $\doeEtaUtil\%$ of the variance in utility recovery but $\doeEtaResid\%$ of the
variance in certified residual ($\eta^2$, one-way by monitor); the same factor that is
irrelevant for trainability is dominant for security. The gain co-design earns its place on the residual,
not on training cost, and a defender who picks the monitor for utility alone, or at random,
gets a certificate that certifies almost nothing. A second task
(Dolly-15k~\cite{dolly2023}, open-domain instruction) replicates the separation:
$\dollyRecovGain$--$\dollyRecovRand\%$ utility recovery across the three monitors at $25\%$
visible, while the certified residual fractions are again $\dollyResidGain$, $\dollyResidQkaug$,
and $\dollyResidRand\%$ ($\eta^2$: $\dollyEtaUtil\%$ utility, $\dollyEtaResid\%$ residual), so
the conclusion is not an artifact of one task.

\begin{table*}[t]
\centering\small
\setlength{\tabcolsep}{4pt}
\begin{minipage}[t]{0.46\textwidth}
\centering
\textbf{(a) Safety--utility DoE}\\[3pt]
\begin{tabular}{@{}lccc@{}}
\toprule
\textbf{Monitor} & \textbf{vis.} & \textbf{util.\ rec.} & \textbf{resid.} \\
\midrule
gain    & .25 & .991 & 0    \\
gain    & .50 & .988 & 0    \\
gain    & .75 & .996 & 0    \\
qk-aug  & .25 & .984 & .43  \\
qk-aug  & .50 & .985 & 0    \\
qk-aug  & .75 & .994 & 0    \\
random  & .25 & .984 & .89  \\
random  & .50 & .985 & .76  \\
random  & .75 & .997 & .60  \\
\bottomrule
\end{tabular}
\phantomsubcaption\label{tab:doe}
\end{minipage}%
\hfill
\begin{minipage}[t]{0.52\textwidth}
\centering
\textbf{(b) Augmentation frontier}\\[3pt]
\begin{tabular}{@{}lrrrr@{}}
\toprule
\textbf{Model} & $b$ & $k_{0.1}$ & $k_{0.01}$ & $\kappa_{0.1}$ \\
\midrule
GPT-2 (MHA)      & 640  & $533^{561}_{516}$ & $638^{638}_{638}$ & 83\% \\
SmolLM2 (MHA)    & 1920 & $1518^{1588}_{1456}$ & $1914^{1918}_{1910}$ & 79\% \\
\addlinespace[2pt]
Qwen-0.5B (GQA)  & 768  & $128^{128}_{128}$ & $128^{128}_{128}$ & 17\% \\
Qwen-1.5B (GQA)  & 1280 & $256^{256}_{256}$ & $256^{256}_{256}$ & 20\% \\
Qwen-3B (GQA)    & 1792 & $256^{256}_{256}$ & $256^{256}_{256}$ & 14\% \\
Qwen-7B (GQA)    & 3328 & $512^{512}_{512}$ & $512^{512}_{512}$ & 15\% \\
Llama-1B (GQA)   & 1920 & $512^{512}_{511}$ & $512^{512}_{512}$ & 27\% \\
Llama-3B (GQA)   & 2816 & $1024^{1024}_{1023}$ & $1024^{1024}_{1024}$ & 36\% \\
\bottomrule
\end{tabular}

\phantomsubcaption\label{tab:frontier}
\end{minipage}

\vspace{6pt}
\begin{minipage}[t]{0.55\textwidth}
\centering
\textbf{(c) Monitor cost and residual ($\approx\!\frac{1}{3}$ blind)}\\[3pt]
\setlength{\tabcolsep}{2.5pt}%
\begin{tabular}{@{}lrrcccc@{}}
\toprule
\textbf{Model} & $\dmodel$ & $\dim\Val$ & gain$_u$ & QK$_u$ & gain$_\sigma$ & QK$_\sigma$ \\
\midrule
GPT-2 (MHA)        & 768  & 768  & .08 & .29 & 0.64 & 3.16 \\
Qwen-0.5B (GQA) & 896  & 128  & 0   & .32 & 0    & 0.34 \\
SmolLM2 (MHA) & 2048 & 2048 & .21 & .72 & 0.76 & 3.64 \\
Qwen-7B (GQA)   & 3584 & 512  & 0   & .08 & 0    & 0.30 \\
\bottomrule
\end{tabular}
\phantomsubcaption\label{tab:multimodel}
\end{minipage}%
\hfill
\begin{minipage}[t]{0.43\textwidth}
\centering
\textbf{(d) Prover cost (Halo2/KZG, BN254)}\\[3pt]
\setlength{\tabcolsep}{3pt}%
\begin{tabular}{@{}lrrr@{}}
\toprule
\textbf{Configuration} & \textbf{Prove} & \textbf{Verify} & \textbf{Proof} \\
\midrule
Typed, full-layer    & $\zkTypedProveLayer$\,s & $\zkTypedVerifyMs$\,ms & $\zkTypedProofKB$\,KB \\
Zero-prod., per-head & $\zkProveHead$\,s & $\zkVerifyHeadMs$\,ms & $\zkProofHeadKB$\,KB \\
Zero-prod., layer    & $\zkProveLayer$\,s & $\zkVerifyMs$\,ms & $\zkProofLayerKB$\,KB \\
Salted commit        & $\zkSaltCommitProve$\,s & $\zkSaltCommitVerifyMs$\,ms & $\zkSaltCommitProofKB$\,KB \\
\bottomrule
\end{tabular}
\phantomsubcaption\label{tab:cost}
\end{minipage}

\caption{Empirical results. \textbf{(a)} Utility recovery (mean, $\utilSeeds$ seeds, rank
$16$) is high and monitor-agnostic; certified residual separates monitors (exact $0$ for
gain, large for random). \textbf{(b)} Monitoring complexity $k_\tau$ =
median$^{q_{75}}_{q_{25}}$ directions for $\sigma_{q+1}/\sigma_1\le\tau$;
$\kappa_{0.1}=k_{0.1}/b$. GQA saturates at $\dim\Val$; MHA needs most of $b$.
\textbf{(c)} Confinement disruption $u$ (the artifact's relative value-path output-change
proxy; not task utility, which \textbf{(a)} measures) and worst-case residual
($\sigma_{k+1}$) at $\approx\!\frac{1}{3}$ blind. Gain is tighter on every model and exact ($0$) on GQA.
\textbf{(d)} Cost is near-flat in rank ($r\!\le\!32$); commit-and-prove hash dominates.}
\label{tab:results}
\end{table*}

\paragraph{The deployed policy: QK plus its optimal in-kernel augmentation.} The deployed
monitor expands the QK operator's coverage within $\ker\pip$ via the
in-kernel augmentation of \cref{thm:augment}, whose ideal-design residual is
$\sigma_{q+1}(\Valbase\Nbasis_\pip)$,
the base value path's gain on the directions actually hidden from that single-head QK monitor
(its exact $\ker\pip$, distinct from the visible-dimension sweep of \cref{tab:multimodel}). On GPT-2 that
gain is $\qkFlagshipResid$ un-augmented and the base path is \emph{full-rank} on $\ker\pip$ (rank
$\qkKerpiRank$; the budget is reader-relative, and this is the \emph{whole-layer} value
reader $\Valbase$, while a single head's value reader caps its budget at the head
dimension $\dhead$), so augmenting $\qkAugDirs$ optimal in-kernel directions cuts it only to
$\qkAugResid$: closing the QK channel is a bounded-residual tradeoff, not free.

\paragraph{The augmentation frontier is architectural, and tracks $n_{kv}$.} The
$\qkFlagshipResid\!\to\!\qkAugResid$ point is one head; the \emph{ideal} frontier
$\sigma_{q+1}(\Valbase\Nbasis_\pip)$ is exact (\cref{thm:augment}; the deployed floor is
its fixed-point compilation's kernel restriction, itself public), so we compute it for
\emph{every} layer and query head of \frontierModels\ checkpoints, spanning two attention
layouts, four families, $0.1$--$7$B, and $Q\!:\!KV$ from $3\!:\!1$ to $8\!:\!1$. We report the
\emph{monitoring complexity} $k_\tau=\min\{q:\sigma_{q+1}/\sigma_1\le\tau\}$, the budget for a
$1/\tau$ attenuation of the base blind read (\cref{tab:frontier}, \cref{fig:frontier}).

The split is architectural and sharp. On the multi-head-attention paths the frontier is
high-rank: the median head needs $\frontierGptKtenMed$ of $\gptBlindDim$ directions on GPT-2
for a $10\times$ attenuation, $100\times$ is unreachable for most heads
($\frontierGptKhundredReach\%$ attain it), and $\kappa_{0.1}=k_{0.1}/b$ is
$\frontierGptKappaPercent\%$ (GPT-2) and $\frontierSmolKappaPercent\%$ (SmolLM2). On every
grouped-query path it collapses at that model's own value-path rank $n_{kv}\dhead$,
with interquartile range zero. That ceiling is algebraic ($\Valbase\Nbasis_\pip$ has at most
$n_{kv}\dhead$ nonzero singular values); what the sweep adds is that it is
\emph{attained} in all six cases, so no cheaper-than-rank ideal monitor exists, and that
$k_{0.1}=k_{0.01}$ there. The fraction is therefore a function of $Q\!:\!KV$, not a constant of
the architecture class: Llama-3.2 at $3\!:\!1$ and $4\!:\!1$ costs
$\frontierLlamaThreeKappaPercent\%$ and $\frontierLlamaOneKappaPercent\%$ where Qwen2.5 at
$6\!:\!1$--$8\!:\!1$ costs $\frontierQwenThreeKappaPercent$--$\frontierQwenOneFiveKappaPercent\%$;
a single-family sweep reports this as one GQA number.

\paragraph{The gain-oracle benchmark, across architectures.} How far could monitor--value
co-design go if coverage were unconstrained? Aligning the
monitor's visible channel to the value path's high-gain subspace strictly dominates QK: the
gain monitor confines at a fraction of QK's disruption cost and bounds the residual
$\gainVsQk\times$ tighter ($\gainResidFactor\times$ below $\sigma_1(W_V)=\wvSigmaMax$ on
GPT-2), on all $\multiModelsCost$ model configurations we test (\cref{tab:multimodel}). Its residual
equals the public base floor $\sigma_{k+1}(W_V)$,
verifier-computable from the base $W_V$ and tunable; a \emph{data}-aligned
monitor confines cheaply but leaves the largest residual, so co-design must target \emph{gain},
not data. Detection stays uninformative on every model (a trained honest LoRA is
$\multiTrainedBlindLo$--$\multiTrainedBlindHi\%$ kernel-active). Grouped-query attention makes
\emph{gain-monitor} closure easier: its value path is inherently low-rank ($\dim\Val =
n_{kv}\dhead = \gqaValueOutDim$), so $\sigma_{k+1}(W_V)$ reaches \emph{zero}: exact closure
against the gain monitor that the high-rank MHA path does not admit. This is closure of the
gain-oracle's channel; the query/key operator's own residual stays nonzero
(\cref{tab:multimodel}), so the two are distinct policies. The pattern holds at deployment
scale: on a 7B GQA model the gain residual is again exactly zero, at single-digit-percent
utility cost, while the operator residual stays nonzero (\cref{tab:multimodel}). The update
infeasibility (\cref{prop:served}(b)) is unchanged.

\paragraph{The prover is cheap (Halo2 + KZG).} We prove the certificate end-to-end on a
Halo2/KZG backend (BN254). The deployed circuit proves the \emph{committed} typed identity
at a post-commitment challenge ($y = Ar - Cm$ with $m = \Mdep\,r$,
\cref{sec:realization}; the committed
factors $A$ and $C$ are the secret witness, and the public output is checked against zero).
The shipped factors are fixed-point-exact ($A = C\Mdep$ on the encoding lattice), so the
honest output is \emph{exactly} zero in the field, not merely small
(\cref{tab:cost}). A tampered adapter that reads a direction hidden from $\Mdep$
yields a nonzero public output ($\zkTypedTamperY$ against the honest exact $0$), so the
verifier's zero-check rejects it. The equivalent zero-product form~\cite{freivalds1979}
$y = B(A(\Nbasis r))$ (used by the named-trigger certificates; same $\zkProbes$-probe
soundness $\approx 2^{-\zkSoundnessBits}$, \cref{lem:allprobes}) measures comparably.
Because the predicate is matmul-only (no execution circuit, no in-circuit rank), it is
cheap (\cref{tab:cost}), and cost is
near-flat in the adapter rank over the measured range: the certificate is the single matrix
product $A - C\Mdep$, whose constraint count is $O(r(\dmodel + \rank\Mdep))$ but which stays
inside the same padded circuit for all $r \le 64$, so measured prover time and proof size rise
only marginally with $r$ (\cref{tab:cost}). A whole checkpoint's \emph{value} adapters are $\pkgRootLayers$ independent per-layer
certificates bound by one package root, which the verifier recomputes and which any swapped
layer changes (\cref{sec:systemcontract}). Across the
$\pkgRootLayers$ GPT-2 value monitors every typed factor passes and every live-channel
substitution is rejected, including the congruent $\tilde A_Z + pE$ attack that defeats a
$\mathrm{GF}(p)$ rank check (a float projection residual, reported only for intuition, separates
honest from substituted by $\serveGuardSepOrders$ orders). Admission is cheap at serving
time: the lattice parse and integer identity take $\guardAdmitMs$\,ms per layer, the
one-time commitment-opening recomputation $\guardOpenS$\,s per layer at load, and the
private sidecar is $\sidecarKB$\,KB per layer. Measured on GPT-2, the salted
commit-and-prove package proves in $\pkgCommitProveS$\,s (plain $\pkgPlainProveS$\,s), verifies in
$\pkgCommitVerifyS$\,s, and ships a $\pkgRootProofKB$\,KB bundle; aggregating the per-layer
commitments into the root costs $\pkgRootCostMs$\,ms, while the in-circuit commitment links
dominate the $\pkgCommitProveS$\,s publication cost. Both the guarantee and these figures cover
the value path, not the Q/K/O or MLP projections; certificate and checkpoint costs sit well
under the minutes a \emph{distilled} GPT-2 inference circuit takes on a comparable
Halo2-based toolchain~\cite{kang2024zkmleurosys}, the trace-proof cost barrier full-model
inference faces~\cite{chantasantitam2026palm}. We also exercise the
\emph{commit-and-prove} binding of \cref{sec:realization} on the measured
commitment~\cite{grassi2021poseidon}: a tampered read factor yields a different $c$
(binding rests on Poseidon's collision resistance; the test exhibits the mechanism, not
the assumption), and re-committing the same adapter under a fresh salt yields a different
$c$, confirming the implemented preimage includes the salt (hiding is the assumption of
\cref{thm:protocol}, not an empirical property this test establishes); the salted per-cert
costs are one-time publisher costs dominated by the in-circuit hash (\cref{tab:cost}). All
measurements are on commodity hardware; the proving environment is in the released artifacts.

\paragraph{Harvested end-to-end certificates.} Two representative public-carrier predicates
are instantiated end-to-end from garak's attack suite~\cite{garak2024} (EZKL/Halo2,
$\garakInjSepOrders$--$\garakGlitchSepOrders$ orders separation between honest and
payload-carrying adapters on GPT-2, $\qwenInjSepOrders$--$\smollmInjSepOrders$ on
Qwen/SmolLM2); the verdict survives int8 quantization and predicts the adapter-mediated
logit effect ($\behavShiftRatio\times$ the confined adapter's shift, \cref{thm:readout}). The
harvested proofs, quantization study, and behavioral validation are in \cref{sec:appendix};
classifying the full attack surface against the cheap/hard boundary is outside this paper.

\section{Related work}
\label{sec:related}

The construction sits at the intersection of four lines of work.

\paragraph{Undetectable backdoors (motivation).} Backdoors can be planted so as to be
provably undetectable from the released weights,
computationally~\cite{goldwasser2022backdoors} or
statistically~\cite{bogdanov2026undetectable}: sparse-hardness and
latent-construction results~\cite{choudhary2026sparseback,eggen2026latentback} and a
detect-to-remove pivot~\cite{goldwasser2025oblivious} reinforce the picture, motivating our
shift from detection to \emph{defense relative to a public, declared monitor}.
Read in reverse, the same algebra is an attack: weight
orthogonalization~\cite{arditi2024refusal} installs $r^{\top}W=0$ against a public refusal
direction $r$.

\paragraph{Reference-based and anomaly-based detection.} PEFTGuard~\cite{peftguard2025} and
per-projection spectral classifiers~\cite{puertolas2026weightspace} learn from labeled
adapter corpora; Neural Cleanse~\cite{neuralcleanse2019} reverse-engineers class-conditioned
candidate triggers; spectral signatures~\cite{spectralsig2018} flag representation-space
outliers in the training data. Each consumes a behavioral, data, or population reference;
absent one, the operator-invisible subclass is provably hard to distinguish, and confinement
supplies the reference structurally.

\paragraph{Structural ZK predicates and execution proofs.} Shang et
al.~\cite{shangchen2026drift} prove norm/rank/sparsity bounds in ZK under the same
publisher--verifier threat model, and FairZK~\cite{fairzk2025} bounds spectral norms for
fairness; an operator-invisible payload passes all such bounds; the object here, a
kernel-confinement relation, is categorically different. The bulk of zkML proves
computational-trace properties~\cite{zkllm2024,zkpot2023,verilora2026}; our certificate
is a single static predicate, not a trace, sidestepping the cost barrier
PAL*M~\cite{chantasantitam2026palm} identifies and the ghost-weights
attack~\cite{hollowllm2026}. Subspace-constrained
adaptation~\cite{subspaceconstrained2026} is a convergent but uncertified empirical
analogue.

\paragraph{Hardware attestation and building blocks.}
PAL*M and Laminator~\cite{laminator2024} root trust in hardware; the audit-gap
analysis~\cite{auditgap2026} formalizes the gap our certificate partially closes. The
construction composes standard
machinery~\cite{kzg2010,legosnark2019,artemis2024,halo2,rangearith2025,dumas2014certificates}
and the ZK-CEC precedent~\cite{shen2026zkcec}. LoRA~\cite{lora2021} is the adapter family;
MasqLoRA~\cite{masqlora2026} and FloatDoor~\cite{floatdoor2026} are live adapter-payload
threats.

\section{Scope, limits, and outlook}
\label{sec:scope}

Each bound below is a theorem or a measured fact.

\paragraph{The base model must be public.} The certificate is cheap because the blind
subspace is a deterministic public function of the base model; a secret base means the
verifier cannot recompute $\Mdep$ in the clear and certification enters the expensive rank
direction (\cref{sec:certificate}). Public and \emph{identified}: the manifest names the
deployed base; relaxing identity to membership in a public candidate set is outside this
paper.

\paragraph{Single-head monitor, extensibly multi-head.} The deployed monitor watches one
head's QK map plus its in-kernel augmentation (\cref{thm:augment}), and the frontier
(\cref{tab:frontier}) is per head; several heads at once is the joint-coverage predicate
(\cref{sec:predicates}), the same typed identity against the stacked monitor, not
evaluated end to end.

\paragraph{Per-layer, not cross-layer.} The certificate confines each layer's adapter within
that layer's monitor. Under an every-layer defense, cross-layer cascades close by the same
rank-nullity argument (\cref{sec:appendix}); under a partial defense, uncertified layers
remain open.

\paragraph{Structural, with a bounded behavioral consequence.} The certificate is a property
of the weights, not a guarantee of behavioral safety: an $\eps$-inert value path shifts any
downstream linear readout by at most $\lVert G\rVert\eps\lVert\delta\rVert$
(\cref{thm:readout}), and a model can pass while unsafe in ways the predicate is silent about.

\paragraph{Serving-endpoint binding.} The certificate and the guard bind the adapter bytes
\emph{admitted} at serving time; binding the admitted tensor to the one the inference kernel
executes is outside the attested boundary, the systems frontier (\cref{sec:systemcontract}).

\paragraph{Knowledge, not just actuation.} Typing restricts the adapter's \emph{read} to
the visible feature $\Mdep x$; it does not constrain the write factor $B$, nor what
behavior is computable from that feature (the full representation is almost surely
injective on prompts~\cite{nikolaou2026injective}, though a low-rank feature need not be).
The boundary is concrete: in the GPT-2 demonstration a directly measured visible contrast
$c^{\top}\Mdep x$ distinguishes the trigger, so a rank-one monitor-typed adapter passes
proof and guard yet is $\knowChanSelectivity\times$ trigger-selective
(\cref{sec:gapcorollary}). This is the joint write-and-read channel named in the outlook, not
a gap in the structural claim (``no operator-invisible channel of this class''); a
knowledge-bounding or behavioral monitor is the complementary defense.

\paragraph{Scope of the empirics.} The augmentation-frontier geometry is measured across
\emph{every} layer and query head of \frontierModels\ checkpoints up to 7B
(\cref{tab:frontier,fig:frontier}), and the task-utility separation replicates across monitor,
budget, rank, and $\utilSeeds$ seeds on two instruction tasks (Alpaca~\cite{taori2023alpaca},
Dolly~\cite{dolly2023}), both modest adaptations on one $0.5$B model.

\paragraph{Outlook.} Monitor-factorized detection is impossible on this channel, and the
reference-free weight statistics we test do not separate benign from malicious use;
construction supplies the declared structural reference that detection lacks. A
proof-carrying adapter is universal
(one static check covers all inputs) and monitor-selective, so monitor quality is what buys
security. What it does not cover, nonlinear, compositional, and data-mediated attacks, is
the frontier: attesting the joint write-and-read channel.

\appendix
\section{Ethics considerations}
\label{sec:ethics}

The construction is defensive: it lets a publisher prove a release free of a known attack
class without disclosing the certified read factor, and lets a consumer or regulator check that
proof. It exposes no new attack. The bounded, named scope of \cref{sec:scope} is part of
the deployment guidance: a passing certificate must be read as the structural,
residual-bounded claim it is, not as a comprehensive safety guarantee.

\providecommand{\serveguardopenscience}{content/open_science_submission}
\section{Open science}
\label{sec:openscience}

Code, formal proofs, and reproduction scripts will be released publicly upon
acceptance for publication. The research artifacts comprise the paper, the
measurement harness, the production pipeline, and the demo. The base models we
measure are public (GPT-2, Qwen2.5-0.5B, SmolLM2-1.7B, Qwen2.5-7B on their
respective hubs); the release will contain code and derived measurements, with
no model weights.

\paragraph{Artifact contents.} (i) The reference semantics and adversarial falsification suite
(clean-accept, OBVS-reject, deficient-basis unsoundness, the defense-specific adaptive
attacker), runnable as a single regression target. (ii) An architecture-agnostic measurement
library that extracts the query/key/value maps of any supported model (fused and separate
projections, multi-head and grouped-query attention) and computes the cost, residual, and
kernel-activity statistics as pure functions, with unit tests. (iii) The end-to-end
zero-knowledge prover (EZKL/Halo2, KZG--BN254), including the commit-and-prove binding and
the garak-harvested certificates. (iv) The Lean~4 development formalizing the soundness
statements of \cref{sec:soundness}, with the paper-to-declaration table of \cref{tab:lean}.

\paragraph{Reproducing the numbers.} Every numeric value in the paper is generated from the
scripts into macro files the paper reads directly, each value carrying a provenance
comment naming its producing script; a documented ordered procedure regenerates the full set (including the
cross-architecture cost table \cref{tab:multimodel} and figure \cref{fig:channel_capacity})
and the canonical CSVs the figures and tables read. The Lean development is sorry-free and
axiom-clean; a committed check enumerates the axioms each cited declaration depends on
(\texttt{propext}, \texttt{Classical.choice}, \texttt{Quot.sound} only). The measurements run
on commodity hardware; the largest model is loaded in reduced precision, and the geometric
statistics are computed in double precision on the extracted weights.

\newcommand{\serveguardartifactstatement}{%
The reference implementation and adversarial suite of \cref{sec:eval} are part of
the research artifact described in \cref{sec:openscience}. The artifact includes
the falsification suite, attack/attestation loop, all measurement scripts, and a
claim-to-test-to-Lean map; the reproduction procedure is in the artifact's playbook.%
}

\section{Machine-checked correspondence}
\label{sec:appendix}

The kernel-algebra and confinement statements of \cref{sec:soundness} are formalized in
Lean~4 (with Mathlib) in a self-contained module. Every cited Lean declaration is sorry-free and
depends only on the standard classical axioms \texttt{propext}, \texttt{Classical.choice},
and \texttt{Quot.sound}. The commit-and-prove and attention-instantiation
results (\cref{thm:protocol,lem:attention}) and the typed encoding bound
(\cref{eq:typedtol}) are standard linear algebra and
cryptographic composition and are \emph{not} mechanized; of the residual-optimality frontier
(\cref{thm:augment}), the exact-certification endpoint (the reader-relative characterization
and its rank lower bound) is mechanized, while the singular-value grading itself is not. The
``machine-checked'' claim is scoped to the algebraic core below. \Cref{tab:lean} maps paper statements to Lean declarations. (The
formalization builds on a general linear-algebra library; the statements and proofs cited
here stand on their own.)

\begin{table}[h]
\centering
\small
\begin{tabular}{@{}ll@{}}
\toprule
\textbf{Paper statement} & \textbf{Lean declaration} \\
\midrule
\eqref{eq:inert} Inertness          & \texttt{Inert} \\
\cref{thm:inertness}(a) Channel-absence & \texttt{attest\_sound} \\
\cref{thm:inertness}(c) Defense           & \texttt{inert\_of\_factor\_through\_visible} \\
\cref{thm:inertness}(b) Factor-through     & \texttt{inert\_iff\_factors\_through} \\
\cref{thm:typed} Typed LoRA          & \texttt{lora\_typed\_inert} \\
\cref{def:served} Served map          & \texttt{served} \\
\cref{lem:served-floor} Base floor    & \texttt{served\_comp\_eq\_base\_of\_update\_inert} \\
\cref{prop:served}(a) Isolation  & \texttt{served\_inert\_iff\_base\_of\_update\_inert} \\
\cref{prop:served}(b) Rank obstruction & \texttt{served\_comp\_ne\_zero\_of\_lowrank} \\
\cref{lem:zeroproduct} Zero-product   & \texttt{inert\_iff\_zero\_product} \\
Joint coverage (cheap dual)           & \texttt{joint\_inert\_iff} \\
\cref{thm:spanning} Deficient-basis counterexample & \texttt{spanning\_necessary} \\
\cref{thm:augment} zero-residual endpoint (certifies iff) & \texttt{readerRelative\_certifies\_iff} \\
\cref{thm:augment} endpoint rank floor & \texttt{readerRelative\_rank\_ge} \\
\cref{thm:tolerance} Encoding tolerance & \texttt{comp\_perturbation} \\
\cref{thm:readout} Behavioral bound   & \texttt{readout\_shift\_le} \\
\bottomrule
\end{tabular}
\caption{Paper-to-Lean correspondence (axiom-clean: \texttt{propext},
\texttt{Classical.choice}, \texttt{Quot.sound} only).}
\label{tab:lean}
\end{table}

\serveguardartifactstatement

\paragraph{Cross-layer cascades under an every-layer defense.} The coverage bound of
\cref{sec:scope} governs cross-layer attacks under a strong premise: if \emph{every}
layer's value path is confined to its own monitor's visible channel, then a coordinated
cancellation cascade (whose per-layer components each lie in that layer's blind
subspace) is inert layer by layer, and the stacked defense closes it by the same
rank--nullity argument, not a new mechanism; a component visible at some layer is outside
the operator-invisible class by definition, and nonlinear propagation between layers is not
modeled beyond this. A cascade against a partially confined stack,
in which some layer is left unmonitored, is outside the certificate, exactly as a single
unconfined layer is.

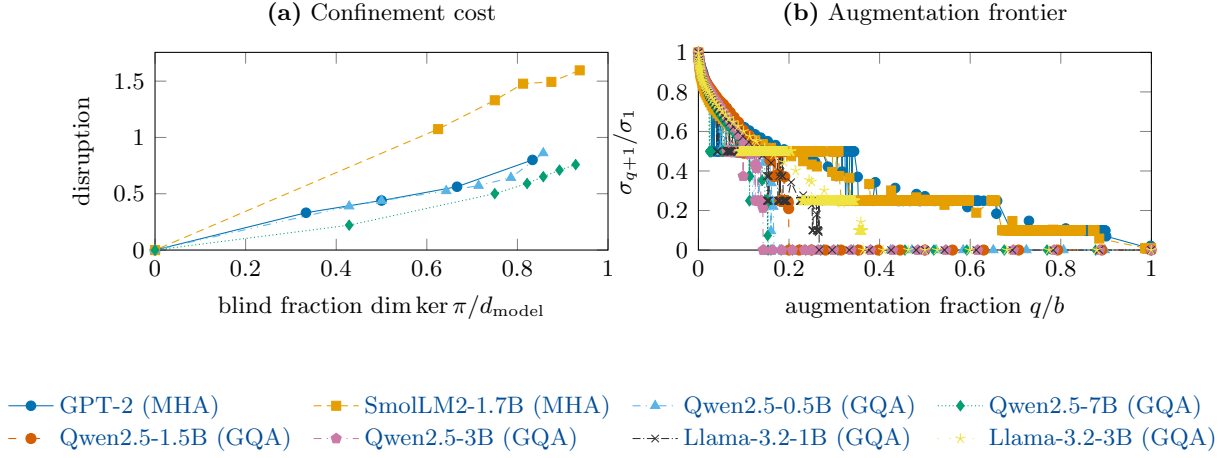
\begin{figure*}[t]
\centering
\ifdefined\arxivversion\begin{adjustbox}{max width=\linewidth}\fi
\begin{tikzpicture}
\begin{axis}[
  name=leftplot,
  width=0.46\textwidth, height=0.255\textwidth,
  title={\footnotesize\textbf{(a)} Confinement cost},
  xlabel={blind fraction $\dim\ker\pip/\dmodel$},
  ylabel={disruption},
  xmin=0, xmax=1, ymin=0,
  cycle list name=oi8, tick label style={font=\footnotesize},
  label style={font=\footnotesize},
]
\addplot table[col sep=comma, x=blind_frac, y=mean_disruption] {data/csv/channel_capacity_gpt2.csv};
\addplot table[col sep=comma, x=blind_frac, y=mean_disruption] {data/csv/channel_capacity_smollm17.csv};
\addplot table[col sep=comma, x=blind_frac, y=mean_disruption] {data/csv/channel_capacity_qwen05.csv};
\addplot table[col sep=comma, x=blind_frac, y=mean_disruption] {data/csv/channel_capacity_qwen7.csv};
\end{axis}
\begin{axis}[
  at={(leftplot.east)}, anchor=west, xshift=12mm,
  width=0.46\textwidth, height=0.255\textwidth,
  title={\footnotesize\textbf{(b)} Augmentation frontier},
  xlabel={augmentation fraction $q/b$},
  ylabel={$\sigma_{q+1}/\sigma_1$},
  xmin=0, xmax=1, ymin=0, ymax=1,
  legend to name=sharedlegend,
  legend columns=4, legend cell align=left,
  legend style={font=\footnotesize, draw=none, /tikz/every even column/.append style={column sep=6pt}},
  cycle list name=oi8, tick label style={font=\footnotesize},
  label style={font=\footnotesize},
]
\addplot table[col sep=comma, x=qfrac, y=rho_med] {data/csv/frontier_median_gpt2.csv};
\addlegendentry{GPT-2 (MHA)}
\addplot table[col sep=comma, x=qfrac, y=rho_med] {data/csv/frontier_median_smollm17.csv};
\addlegendentry{SmolLM2-1.7B (MHA)}
\addplot table[col sep=comma, x=qfrac, y=rho_med] {data/csv/frontier_median_qwen05.csv};
\addlegendentry{Qwen2.5-0.5B (GQA)}
\addplot table[col sep=comma, x=qfrac, y=rho_med] {data/csv/frontier_median_qwen7.csv};
\addlegendentry{Qwen2.5-7B (GQA)}
\addplot table[col sep=comma, x=qfrac, y=rho_med] {data/csv/frontier_median_qwen2515b.csv};
\addlegendentry{Qwen2.5-1.5B (GQA)}
\addplot table[col sep=comma, x=qfrac, y=rho_med] {data/csv/frontier_median_qwen253b.csv};
\addlegendentry{Qwen2.5-3B (GQA)}
\addplot table[col sep=comma, x=qfrac, y=rho_med] {data/csv/frontier_median_llama321b.csv};
\addlegendentry{Llama-3.2-1B (GQA)}
\addplot table[col sep=comma, x=qfrac, y=rho_med] {data/csv/frontier_median_llama323b.csv};
\addlegendentry{Llama-3.2-3B (GQA)}
\end{axis}
\node[below=6mm] at (current bounding box.south) {\ref{sharedlegend}};
\end{tikzpicture}
\ifdefined\arxivversion\end{adjustbox}\fi
\caption{\textbf{(a)} QK-monitor confinement cost grows with blind fraction, across the four
checkpoints with capacity measurements. \textbf{(b)} Ideal-design augmentation frontier
$\sigma_{q+1}/\sigma_1$ (median
over all layers and heads, \cref{thm:augment}). All eight checkpoints are plotted;
the four in \textbf{(a)} keep their colours. The six GQA paths (two families, $Q\!:\!KV$ from
$3\!:\!1$ to $8\!:\!1$) fall to zero at their own value-path rank $n_{kv}d_{\mathrm{head}}$, so
they collapse at different fractions of $b$; the MHA paths (GPT-2, SmolLM2) decay only near
full coverage.}
\label{fig:channel_capacity}
\label{fig:frontier}
\end{figure*}

\section{The proof-carrying-adapter protocol and trust boundary}
\label{sec:contract}

We state the deployed system as a named protocol so downstream work cites a fixed interface.
Over a public base \emph{checkpoint} $W_0$ with value map $\Valbase$:
\begin{description}
\item[\normalfont$\mathsf{CompileMonitor}(W_0,\mathsf{policy})\to(\Mdep,\mu)$]
  deterministically builds the finite-precision monitor $\Mdep$ from the full base checkpoint
  (the QK map plus its optimal in-kernel augmentation, \cref{sec:defense}) and an authenticated
  \emph{manifest} $\mu$ (base hash, model/layer/head, policy, quantization, the compiler
  digest pinning the SVD implementation with its ordering, sign, and rounding conventions,
  monitor digest $h_{\Mdep}$, verification-key digest, and the public residual floor).
\item[\normalfont$\mathsf{TypeTrain}(\Mdep,\mathsf{data})\to(B,A,C)$] trains the adapter
  monitor-typed, $A=C\Mdep$, so $\Valupd=BA$ is inert for every $B$ (\cref{thm:typed}).
\item[\normalfont$\mathsf{CommitCert}(A,C)\to\commit$] one salted commitment
  $\commit=\mathrm{Com}(A,C;\rho)$ that the certificate proof opens, hiding $A$ and $C$ from a
  verifier who never sees the certified read factor, fixed before the challenge is derived. We bind this
  \emph{same} $\commit$ into the package (below), so there is no second, unlinked shipped-factor
  commitment; committing the read factor \emph{alone}, which would let a consumer open with a
  read-only witness, is not expressible under the stock toolchain's global commitment
  visibility, so opening $\commit$ uses $(C,\rho)$, which $\mathsf{Package}$ delivers to the
  weight-holding consumer as a private serving sidecar; the remote verifier receives neither.
\item[\normalfont$\mathsf{Package}(B,A,\mu,\commit,C,\rho)\to(h_{\mathrm{pkg}},\mathcal P)$]
  serializes the shipped package with root $h_{\mathrm{pkg}}=H\bigl(\mu,\,
  (\ell,\mathsf{tensorID}_\ell,\mathsf{shape}_\ell,H(B_\ell),\commit_\ell)_{\ell=1}^{L}\bigr)$,
  binding the shipped write factor and the \emph{proof's own} commitment per layer;
  $\mathcal P$ carries the checkpoint $(B,A)$ and, to the serving consumer only, the
  sidecar $(C,\rho)$.
\item[\normalfont$\mathsf{Prove}(\mathsf{st};A,C,\rho)\to\pi$] the certificate proof of the typed
  relation $A=C\Mdep$, opening $\commit$ (\cref{sec:realization}).
\item[\normalfont$\mathsf{Verify}(\mathsf{st},\pi)\to\{0,1\}$] one layer's proof, against
  $\mathsf{st}=(h_{\mathrm{base}},h_{\mathrm{pkg}},\mu,\Mdep\text{ or its digest},\commit,
  \dmodel,\dhead,r,\ldots)$; $\mathsf{VerifyBundle}$ runs $\mathsf{Verify}$ on all $L$ layers and
  recomputes $h_{\mathrm{pkg}}$.
\item[\normalfont$\mathsf{ServeGuard}(W_0,\mathsf{policy},\commit,\tilde A,C,\rho)\to\{0,1\}\times\text{tensor}$]
  the serving-time guard, run by the weight-holding consumer on the served bytes and the
  private sidecar $(C,\rho)$: (i) it parses $\tilde A$ as a \emph{canonical lattice tensor},
  rejecting any $\tilde A\neq\tilde A_Z\,2^{-2F}$ for
  $\tilde A_Z=\mathrm{round}(\tilde A\,2^{2F})$ (a float near a lattice point but not on it
  is refused, so no off-lattice perturbation, however small, is admitted for the
  unrestricted write factor to amplify); (ii) it verifies the commitment opening
  $\commit=\mathrm{Com}(\tilde A,C;\rho)$, binding the served bytes to the \emph{certified}
  factor; (iii) it checks the typed identity exactly over the integer lattice,
  $\tilde A_Z=C_Z(\Mdep)_Z$ with $C_Z=\mathrm{round}(C\,2^{F})$ and
  $(\Mdep)_Z=\mathrm{round}(\Mdep\,2^{F})$ (the honest adapter is built as $A=C\Mdep$ on
  the lattice, so this is an integer equality), where the monitor is \emph{recomputed} from
  the authenticated base by $\mathsf{CompileMonitor}$, not a shipped matrix (closing the
  deficient-monitor substitution); and (iv) it returns the canonical decoded tensor
  $\tilde A_Z\,2^{-2F}$, the tensor inference must execute (a serving stack that re-casts
  it to a floating format re-enters the executed-tensor frontier, \cref{sec:scope}).
  Given the opening, the admission check is \emph{exact and unconditional}: no prime, no
  field check, no acceptance threshold, hence no slack that an unbounded write factor $B$
  could amplify (\cref{thm:typed}). Exactness is a theorem about the canonical weight
  relation; floating-point evaluation error inside the inference matmul is outside it
  unless the kernel implements the fixed-point semantics or a separate numerical-error
  bound is applied. Checking $\rank$ over a fixed
  prime field instead would be unsound, since a congruent $\tilde A_Z+pE$ shares the $\mathrm{GF}(p)$ rank yet carries
  a live channel over the integers; the integer identity above rejects it. The guard reads
  only the public monitor, the committed read factor, and the served weights, and a defensive range
  bound on $\tilde A_Z$ rejects wild factors (exactness does not depend on it).
\end{description}

The guarantee composes three levels. \emph{Algebraic:} for every accepted layer, $A=C\Mdep
\Rightarrow \forall\delta\in\ker\Mdep,\ BA\delta=0$ (\cref{thm:typed}). \emph{Served-checkpoint:}
with $\Valsrv=\Valbase+BA$, $\forall\delta\in\ker\Mdep,\ \Valsrv\delta=\Valbase\delta$
(\cref{lem:served-floor}): the publisher adds no blind-subspace channel and the residual is the
authenticated base floor. \emph{Computational, in two parts.} \emph{Typed-relation soundness}
(\cref{thm:protocol}): except with $\eps_{\mathrm{bind}}+\eps_{\mathrm{KS}}+q_H\lvert
S\rvert^{-\zkProbes}$, the read factor committed under $\commit$ is monitor-typed, so the
update it forms with any write factor is inert, a privacy-preserving attestation of the committed adapter, checkable without the read factor.

\paragraph{No-wrap ledger.} The circuit input is $m'=(m,0)$: the appended coordinate is a
\emph{fixed public zero} multiplying the committed salt column of $C$, so the salt enters
the commitment preimage but never the identity, and the checked relation stays the
homogeneous $A_Z\,r = C_Z\,m$. The verifier computes $m=(\Mdep)_Z\,r$ canonically in
exact integer arithmetic outside the circuit (no wrap can occur there) and it enters as a
range-checked public input. Every in-circuit cell is bounded by the lookup decomposition
(base $2^{14}$, two legs): magnitude below $2^{\nwCellBits}$, a circuit \emph{constraint}
holding for any witness, not an observed maximum, at the maximum width configured over the
twelve generated layer circuits:
\begin{center}\small
\begin{tabular}{@{}lrr@{}}
\toprule
\textbf{Intermediate} & \textbf{terms} & \textbf{enforced bound (bits)} \\
\midrule
$A_Z\,r$ & $\dmodel$ & $\nwBoundAr$ \\
$C_Z\,m'$ & $s{+}1$ & $\nwBoundCm$ \\
$A_Z\,r - C_Z\,m'$ & -- & $\nwBoundDiff$ \\
\bottomrule
\end{tabular}\end{center}
Minimum margin: $\nwMarginWorst$ bits below $p/2$ (BN254).
\emph{Served-integrity}: $\mathsf{ServeGuard}$ admits $\tilde A_\ell$ only if
$\tilde A_\ell=C_\ell\Mdep[\ell]$ for the committed $C_\ell$, an exact integer check that both types
the served factor and binds it to the commitment, so the \emph{served} update $B\tilde A_\ell$ is
confined for the served $B$ whatever an intermediary shipped. The composition is the system's contract.

The contract this interface supports, \cref{thm:endtoend}, is stated and proved in
\cref{sec:systemcontract}.

The probabilistic term in \cref{thm:protocol} is the following all-probes bound on the
probed identity of \cref{sec:realization}.

\begin{lemma}[All-probes soundness]
\label{lem:allprobes}
Let $E := A - C\Mdep$ be the residual on the encoding lattice, fixed by the commitment
$\commit$ and the public statement $\mathsf{st}$, and suppose $E \neq 0$. If the probe
vectors $r_1,\dots,r_{\zkProbes}$ have coordinates drawn independently and uniformly from
the bounded integer range $S$ after $\commit$ is fixed, and the no-wrap range checks hold,
then $\Pr[\,E r_j = 0 \text{ for all } j\,] \le \lvert S\rvert^{-\zkProbes}$, and a prover
making $q_H$ random-oracle queries to steer the probes succeeds with probability at most
$q_H\,\lvert S\rvert^{-\zkProbes}$.
\end{lemma}

\begin{proof}
Fix a nonzero row $e$ of $E$. Conditioning on all coordinates of $r_j$ except one at which
$e$ is nonzero, at most one value of the remaining coordinate zeroes $e \cdot r_j$, so
$\Pr[e \cdot r_j = 0] \le 1/\lvert S\rvert$; independence across $j$ gives
$\lvert S\rvert^{-\zkProbes}$. The union bound over the $q_H$ salt choices gives the
grinding term. The range checks make a field zero an integer zero, so the field identity
checked in-circuit implies the integer one.
\end{proof}

\paragraph{Overheads.} The typing is a training-time parameterization, not a serving-time
cost: $\mathsf{TypeTrain}$ constrains the read factor to $A=C\Mdep$ during training, and the
shipped $(B,A)$ is an \emph{ordinary} LoRA once materialized, so \emph{inference overhead is
zero} (the served model runs the standard adapter). The added artifacts are small and
one-time: the read factor $A$ is $r\times\dmodel$ (a rank-$8$ GPT-2 value adapter is
$\gptLoraSecretRankEight$ secret entries), the manifest $\mu$ is a handful of hashes and
scalars, and the whole-checkpoint proof bundle with its package root is
$\pkgRootProofKB$\,KB, verified in $\pkgCommitVerifyS$\,s (the root itself
$\pkgRootCostMs$\,ms). Prover time is the only material cost and is one-time at publication
($\pkgCommitProveS$\,s for a GPT-2 checkpoint).

\begin{table}[h]
\centering\small
\caption{Trust boundary of the deployed certificate.}
\label{tab:trust}
\begin{tabular}{@{}ll@{}}
\toprule
Object / party & Status \\
\midrule
Adapter publisher and prover & untrusted \\
Read factor $A$, witness $C$ & secret, adversarial \\
Write factor $B$ & shipped in the clear, hash-bound \\
Public base hash, monitor policy & trusted public statement \\
Monitor compiler $\mathsf{CompileMonitor}$ & deterministic, reproducible \\
SNARK, commitment, SRS & cryptographic assumptions \\
Artifact distribution & bound by $h_{\mathrm{pkg}}$ (write factor $+$ the proof's commitment) \\
Distribution intermediary (may swap $A$) & untrusted; defeated by $\mathsf{ServeGuard}$ \\
Served read factor $\tilde A$ & admitted only if $\mathsf{ServeGuard}$ passes (public monitor) \\
Serving endpoint & runs $\mathsf{ServeGuard}$; full behavior out of scope absent TEE \\
Served model behavior & not certified \\
\bottomrule
\end{tabular}
\end{table}

\section{Knowledge boundary and parameter-efficiency corollary}
\label{sec:gapcorollary}

\noindent The certificate excludes blind actuation; it does not constrain visible-channel
knowledge or visible writes, and no predicate keyed to a fixed public carrier decides what
rides directions outside that carrier. The demonstration below makes this boundary concrete
on GPT-2; the general classification of what cheap certificates can and cannot decide is
outside this paper.

\paragraph{Boundary demonstration: knowledge without blind actuation.} The
knowledge-vs-actuation boundary of \cref{sec:scope} is concrete on real GPT-2: the rank-1
typed adapter $A=c^{\top}\Mon$ (with $c$ the monitor-visible contrast of a trigger prompt)
passes both the certificate and $\mathsf{ServeGuard}$ and reads nothing from $\ker\Mon$
(machine zero), yet its output-logit shift is $\knowChanSelectivity\times$ larger on the
trigger than on controls. The certificate behaves exactly as specified (it verifies a
genuinely typed adapter); the knowledge-conditioned write rides the \emph{visible}
directions: the construction is \emph{inside} the certified class but outside the
blind-actuation class the certificate excludes, which is why a knowledge-bounding
or behavioral monitor is the complementary defense.

\paragraph{Parameter-efficiency corollary.} Confining the value
path to a $k$-dimensional visible channel caps the adapter's rank at $k$, so it admits a
smaller $r\times k$ read parameterization in place of $r\times\dmodel$. Empirically a rank-%
$\poneCheapRankLo$ adapter confined to a quarter-width channel ($\poneCheapParamRatio\times$
fewer parameters than the unconfined baseline) still recovers $\poneCheapRecoveryLo\%$ of
the task gain.

\section{Zero-product form and auxiliary theorems}
\label{sec:zeroproduct}

\paragraph{Optimal in-kernel augmentation (\cref{thm:augment}).}
\begin{proof}[Proof idea]
$\Valbase P_{\ker\pip} = (\Valbase\Nbasis_\pip)\Nbasis_\pip^{\top}$; monitoring $T$ removes the
corresponding $q$ domain directions of $\Valbase\Nbasis_\pip$, and by the Courant--Fischer
min--max characterization the operator norm of the restriction to the remaining $(b{-}q)$
directions is minimized by dropping its top-$q$ right-singular directions, leaving its
$(q{+}1)$-th singular value.
\end{proof}

\paragraph{Ideal-monitor tolerance.} A deployer who instead wishes to
name an \emph{ideal} real-valued monitor $\Mon$ and serve its fixed-point compilation
$\Mdep=\Mon_{\mathrm{fp}}$ pays an explicit, optional tolerance: for any real payload
$\delta\in\ker\Mon$,
\begin{equation}
\Valupd\,\delta = BC(\Mon_{\mathrm{fp}}-\Mon)\,\delta, \quad
\lVert \Valupd\,\delta\rVert \le \lVert BC\rVert\,\lVert \Mon_{\mathrm{fp}}-\Mon\rVert\,
\lVert\delta\rVert,
\label{eq:typedtol}
\end{equation}
since $\Mon\delta = 0$. Here $\lVert\Mon_{\mathrm{fp}}-\Mon\rVert$ is the public encoding error of the
compiled monitor at the stated fixed-point resolution, and $\lVert BC\rVert \le \lVert
B\rVert\,\lVert C\rVert$ is bounded by the range-checked witness $C$ and the norm of the
\emph{publicly shipped, package-bound} write factor $B$ (\cref{sec:realization}), not by any
in-circuit constraint on $B$. This ideal-monitor reading is a modeling choice a deployer may
state; the deployed guarantee against $\Mdep$ needs none of it.

\begin{lemma}[Zero-product certificate]
\label{lem:zeroproduct}
If $\img\Nbasis = \ker\Mon$ then $\Inert(\Mon,\Val) \Leftrightarrow \Val\circ\Nbasis = 0$.
If only $\ker\Mon\subseteq\img\Nbasis$, then $\Val\circ\Nbasis=0$ implies
$\Inert(\Mon,\Val)$.
\end{lemma}

\begin{theorem}[Deficient basis: a counterexample]
\label{thm:spanning}
If $\img\Nbasis \subsetneq \ker\Mon$ (strict), there is a value map $\Val$ with
$\Val\circ\Nbasis = 0$ yet $\neg\,\Inert(\Mon,\Val)$: a deficient basis admits a
non-inert map that passes the certificate. The verifier must check
$\Mon\Nbasis = 0$ \emph{and} $\rank\Nbasis = \dmodel-\rank\Mon$ in the clear.
\end{theorem}

\noindent Training can enforce the typed form $A = C\Mon$ by parameterization, or penalize
the read factor's blind component directly as a regularizer; an exact typed adapter has
zero blind component.

\section{Predicate family and harvested certificates}
\label{sec:predicates}

\begin{table}[h]
\centering
\small
\begin{tabular}{@{}p{2.3cm}p{5.0cm}p{4.6cm}@{}}
\toprule
\textbf{Predicate} & \textbf{Statement} & \textbf{Role} \\
\midrule
\textbf{Inertness} & $\ker\Mon \subseteq \ker\Val$, certified by $A = C\Mon$ (or $\Val\,\Nbasis = 0$) & the defense; absence of the OI channel \\
\textbf{Confinement} & $\Inert(\Mon, \Val\circ\Pivis)$ for any $\Val$ & inertness by construction \\
\textbf{Joint coverage} & $\Inert(\Mon_1\!\times\!\Mon_2, \Val)$ & cheap dual \\
\textbf{Behavioral bound} & $\lVert G\Val\delta\rVert \le \lVert G\rVert\,\eps\,\lVert\delta\rVert$ for $\delta\in\ker\Mon$ & $\eps$-inertness bounded shift \\
\textbf{Named trigger} & $T \subseteq \ker\Val$, certified by $\Val\,\Nbasis=0$ with $\img\Nbasis=T$ & fixed public trigger \\
\bottomrule
\end{tabular}
\caption{The defense's predicate family (\cref{sec:defense}). The exact predicates arise
from public homogeneous identities; the behavioral row is a norm consequence of inertness.}
\label{tab:predicates}
\end{table}

\paragraph{Cheap dual (coverage).} Inertness against the joint monitor $\Mon_1\times\Mon_2$
is inertness against $\ker\Mon_1\cap\ker\Mon_2$, which is publicly computable when both
monitors are public; this certifies the dual property (closing two monitors' shared blind
spot) by the same cheap zero-product, with no rank lower bound. A budget on a \emph{secret}
monitor is a rank lower bound, the hard direction, and is out of scope (\cref{sec:scope}).

\paragraph{Harvested proofs.} From garak's glitch probe: \garakGlitchTokens\ single-token
triggers (rank-\garakGlitchRank\ subspace), residual $\garakGlitchCleanResid$ (accept) vs
$\garakGlitchBackdoorResid$ (reject), $\garakGlitchSepOrders$ orders apart. Injection split:
$\garakInjBlindFrac$ concealable in $\ker\pip$, $\garakInjSepOrders$ orders separation.
Architecture-independent: $\qwenInjSepOrders$ (Qwen), $\smollmInjSepOrders$ (SmolLM2).
Int8 quantization keeps the tolerance verdict separated ($\quantIntEightHeadroom$ orders
below the backdoor); int4 erodes the margin to $\quantIntFourMarginLo\times$.

\paragraph{Behavioral validation.} The rejected adapter shifts glitch-token logits
$\behavShiftRatio\times$ more than the confined one ($\behavRatioMean\pm\behavRatioStd\times$
over $\behavSeeds$ draws); the confined adapter's shift stays within the
\cref{thm:readout} bound. The certificate prevents an adapter from weaponizing the glitch
subspace; it does not fix the base model (which fails $\garakBaseGlitchFail\%$ of
$\garakBaseGlitchProbes$ prompts on its own).

\section*{AI use statement}
Large language model tooling was used, under the authors' direction and review, to assist with the
Lean~4 formalization of the soundness theorems (every cited Lean declaration is sorry-free and
axiom-clean; results or clauses carried only by paper proofs are identified explicitly in
Appendix~A; the human authors are responsible for the statements and their correspondence to the
paper), with implementation
of the zero-knowledge prover integration and the measurement harness, and with drafting and editing
of the manuscript. All claims, proofs, and reported measurements were verified by the authors, who
take full responsibility for the paper's content. No generative model was used to produce or
fabricate experimental data.

\bibliographystyle{plainnat}
\bibliography{refs}

\end{document}